%% file: main.tex
\documentclass[acmtog,screen,nonacm]{acmart}

\usepackage{booktabs}

\usepackage[ruled]{algorithm2e}

\SetAlFnt{\small}
\SetAlCapFnt{\small}
\SetAlCapNameFnt{\small}
\SetAlCapHSkip{0pt}

\authorsaddresses{}

\input{definitions}
\input{preamble}

\begin{document}
\title{Deep Medial Fields}

\input{authors}

\begin{abstract}
    \input{sec/0_abstract}
\end{abstract}

\input{fig/teaser}
\maketitle
\input{sec/1_introduction}
\input{sec/2_related}
\input{sec/3_method}
\input{sec/4_applications}
\input{sec/5_implementation}
\input{sec/6_conclusions}
\input{sec/X_appendix}

\bibliographystyle{ACM-Reference-Format}
\bibliography{egbib}
\end{document}

%% file: definitions.tex
\usepackage{dsfont}

\newcommand{\Real}{\mathbb{R}}
\newcommand{\Shape}{\mathcal{O}}
\newcommand{\Boundary}{\partial\mathcal{O}}
\newcommand{\MAxis}{\mathcal{M}}

\newcommand{\Line}{\mathbf{s}}
\newcommand{\SDF}{\Phi}
\newcommand{\USDF}{\mathbf{d}}
\newcommand{\MField}{M} 
\newcommand{\Project}[1]{\mathrm{proj}_{#1}}
\newcommand{\Loss}[1]{\mathcal{L}_\textrm{#1}}

\DeclareMathOperator*{\argmax}{arg\,max}
\newcommand{\IE}{\mathds{E}}
\newtheorem{prop}{Proposition}[section]

%% file: preamble.tex
\usepackage{enumitem}
\usepackage{microtype}

\usepackage{color}
\definecolor{turquoise}{cmyk}{0.65,0,0.1,0.3}
\definecolor{purple}{rgb}{0.65,0,0.65}
\definecolor{dark_green}{rgb}{0, 0.5, 0}
\definecolor{orange}{rgb}{0.8, 0.6, 0.2}
\definecolor{red}{rgb}{0.8, 0.2, 0.2}
\definecolor{darkred}{rgb}{0.6, 0.1, 0.05}
\definecolor{blueish}{rgb}{0.0, 0.3, .6}
\definecolor{light_gray}{rgb}{0.7, 0.7, .7}
\definecolor{pink}{rgb}{1, 0, 1}
\definecolor{greyblue}{rgb}{0.25, 0.25, 1}

\usepackage{soul}
\usepackage{multirow}
\usepackage{overpic}

\newcommand{\CIRCLE}[1]{\raisebox{.5pt}{\footnotesize \textcircled{\raisebox{-.6pt}{#1}}}}

\newcommand{\Figure}[1]{Figure~\ref{fig:#1}}

\newcommand{\Table}[1]{Table~\ref{tab:#1}}
\newcommand{\eq}[1]{\eqref{eq:#1}}

\newcommand{\Section}[1]{Section~\ref{sec:#1}}
\newcommand{\Appendix}[1]{Appendix~\ref{app:#1}}

%% file: authors.tex
\author{Daniel Rebain}
\affiliation{
    \institution{University of British Columbia}
    \city{}
    \country{}
}
\author{Ke Li}
\affiliation{
    \institution{Simon Fraser University, Google Research}
    \city{}
    \country{}
}
\author{Vincent Sitzmann}
\affiliation{
    \institution{MIT CSAIL, Stanford University}
    \city{}
    \country{}
}
\author{Soroosh Yazdani}
\affiliation{
    \institution{Google Research}
    \city{}
    \country{}
}
\author{Kwang Moo Yi}
\affiliation{
    \institution{University of British Columbia}
    \city{}
    \country{}
}
\author{Andrea Tagliasacchi}
\affiliation{
    \institution{Google Research, University of Toronto}
    \city{}
    \country{}
}

%% file: sec/0_abstract.tex
Implicit representations of geometry, such as occupancy fields or signed distance fields~(SDF),  have recently re-gained popularity in encoding 3D solid shape in a functional form.
In this work, we introduce medial fields: a field function derived from the medial axis transform (MAT) that makes available information about the underlying 3D geometry that is immediately useful for a number of downstream tasks.
In particular, the medial field encodes the local thickness of a 3D shape, and enables O(1) projection of a query point onto the medial axis.
To construct the medial field we require nothing but the SDF of the shape itself, thus allowing its straightforward incorporation in any application that relies on signed distance fields.
Working in unison with the O(1) surface projection supported by the SDF, the medial field opens the door for an entirely new set of efficient, shape-aware operations on implicit representations.
We present three such applications, including a modification to sphere tracing that renders implicit representations with better convergence properties, a fast construction method for memory-efficient rigid-body collision proxies, and an efficient approximation of ambient occlusion that remains stable with respect to viewpoint variations.

%% file: fig/teaser.tex
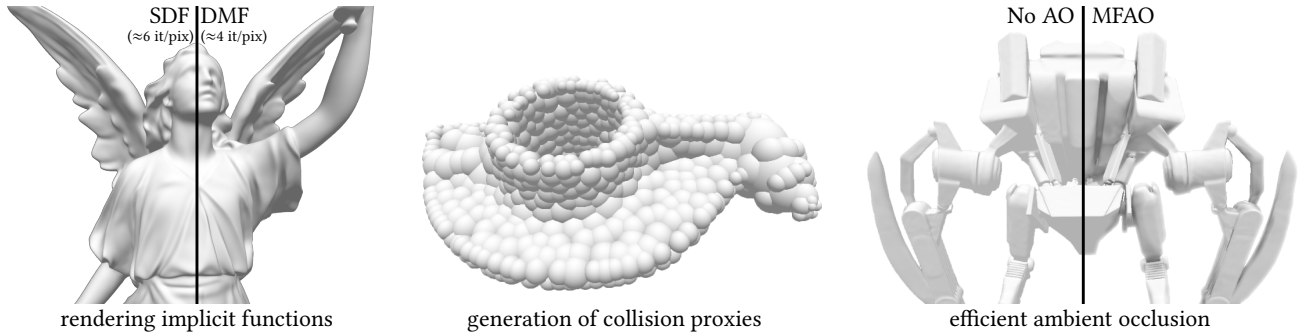
\begin{teaserfigure}
\centering
\setOverpic{grid}
\begin{overpic} [width=0.33\linewidth] {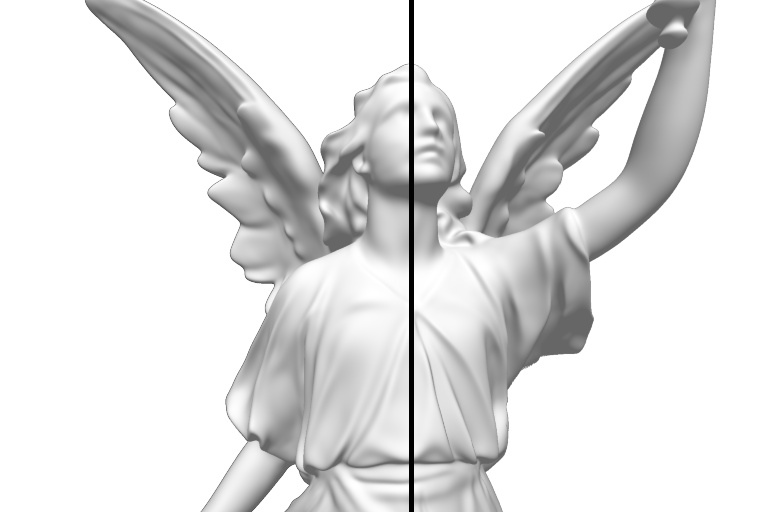}
\put(43,63){SDF}
\put(38,59){{\scriptsize($\approx$6 it/pix)}}
\put(54.5,63){DMF}
\put(54.5,59){{\scriptsize($\approx$4 it/pix)}}
\put(23,-5){rendering implicit functions}
\end{overpic}
\begin{overpic} [width=0.33\linewidth, trim = 50 60 50 0, clip] {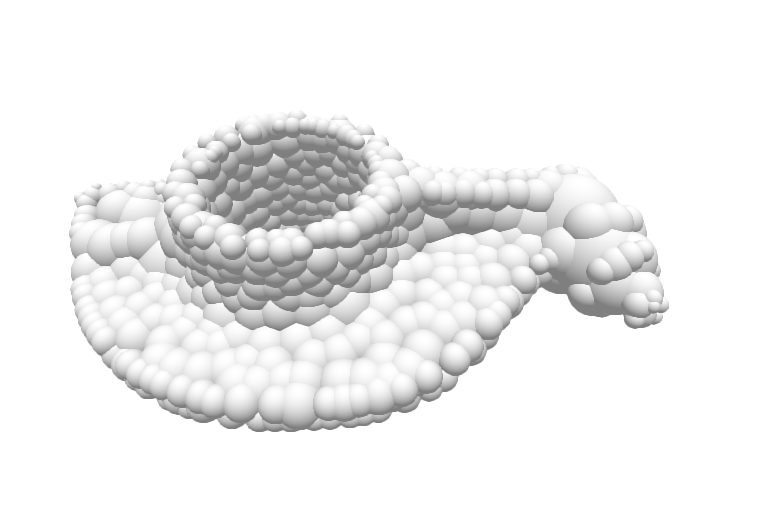}
\put(13,-5){generation of collision proxies}
\end{overpic}
\begin{overpic} [width=0.33\linewidth] {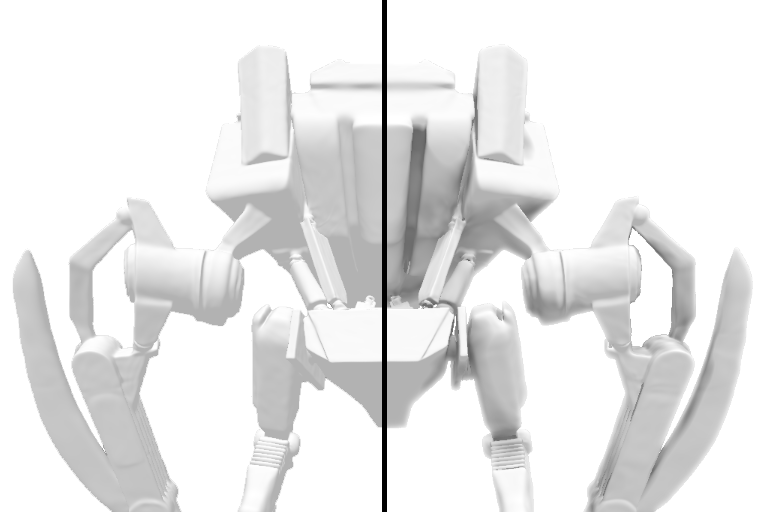}
\put(33,63){No AO}
\put(52,63){MFAO}
\put(20,-5){efficient ambient occlusion}
\end{overpic}
\vspace{-.5em}
\caption{
Deep medial field implicitly encode information related to the medial axis in a neural field function. Such a function can be used to accelerate iterative ray tracing of implicits  (left), quickly compute physics acceleration proxies (middle), or implement a new form of ambient occlusion (right).
}
\label{fig:teaser}
\end{teaserfigure}

%% file: sec/1_introduction.tex
\section{Introduction}
Neural implicit representations of 3D geometry and appearance have recently emerged as an attractive alternative to conventional discrete representations such as polygonal meshes or grids.
For example, \citeN{mescheder2019occupancy} and \citeN{park2019deepsdf} store geometry of a 3D object as occupancy and signed distance field, while \citeN{sitzmann2019srns} and \citeN{mildenhall2020nerf} additionally infer appearance via differentiable rendering.
These new neural representations enable spatially adaptive and high-resolution shape representations of 3D signals, while allowing the learning of priors over shape~\cite{chen2019learning} and appearance~\cite{Oechsle2019texture,sitzmann2019srns}.
In this work, we investigate a novel representation of 3D geometric information that gives O(1) access to quantities immediately useful to a variety of downstream tasks.

\newcommand{\footnoteone}{Note that the concept of local thickness is well defined
\textit{both} inside the shape (the local thickness of the object), as well as outside (the local thickness of the ambient space).}
\newcommand{\footnotetwo}{An empty sphere is a sphere that does not intersect the object's boundary; hence it either lies completely inside or completely outside the object.}
\paragraph{Medial field}
Specifically, we introduce the analytical concept of \textit{medial fields}.
The medial field, as in the case of occupancy and SDF, is a scalar function defined over $\Real^d$.
While signed distance functions retrieve the radius of an \textit{empty}\footnote{\footnotetwo} sphere centered at $x$ that is tangent to the surface of an object, medial fields query the \textit{local thickness}\footnote{\footnoteone} at the query point $x$.
Local thickness expresses the size of a sphere that is \textit{related} to the one retrieved by the SDF: it is tangent to the object surface at the same location, it fully contains the SDF sphere, and it is the \textit{largest} empty sphere satisfying these properties.
However, assuming we can easily and efficiently retrieve it with~O(1) complexity,~\textit{``What applications could it enable?''}

\paragraph{Applications}
In this paper we explore a (likely incomplete) portfolio of applications for medial fields.
First and foremost, we look at the classical problem of rendering implicit functions, and realize that the popular ``sphere tracing'' algorithm proposed by~\citeN{hart1996sphere} relies on iteratively querying empty spheres provided by the signed distance function.
Conversely, by exploiting medial fields, we can compute spheres that are larger, resulting in fewer, larger steps, and hence an overall improved convergence rate; see~\Figure{teaser}~(left)
We then consider real-time physics simulation, for which the efficient computation of collision proxies is of central importance, and that typically involve classical algorithms executed on polygonal meshes~\cite{ericson2004real}.
We instead propose a solution that extracts a volumetric approximation as a collection of spheres, where the locally largest empty spheres queried by the medial sphere result in an efficient coverage of space without using an excessive number of proxies; see~\Figure{teaser}~(middle).
Finally, we employ local thickness as a way to efficiently identify deep creases in the surface of objects where light is unlikely to reach, hence providing a method to approximate ambient occlusion shading effects within the realm of implicit shape representations; see~\Figure{teaser}~(right)
While these applications showcase the versatility of medial fields, we must ask the question: \textit{``How can we compute (and store) medial fields?''} 

\paragraph{Computation}
In computational geometry, the concept of local thickness is formalized by the Medial Axis Transform~\cite{blum}, a dual representation of solid geometry from which we derive the nomenclature ``medial field''.
Medial axis computation has been studied extensively in computer science, but the design of algorithms as effective as for other types of transforms~(e.g. the Fourier Transform) has been elusive; the problem is particularly dire in~$\Real^3$, where the computation of the medial surfaces from polygonal meshes essentially remains an open challenge~\cite{survey}.
We overcome these issues by expressing the problem of medial field computation as a continuous constrained satisfaction problem, which we approach 
by stochastic optimization, and that only requires querying the signed distance field.
Further, rather than relying on discrete grids to store medial fields, whose memory requirements are excessive at high resolutions, we compactly store the medial field function within the parameters of a neural network.

\noindent
To summarize, in this paper we introduce:
\vspace{-.5em}
\begin{itemize}[leftmargin=*]
\setlength\itemsep{0em}
\item ``medial fields'' as a way to augment signed distance fields with meta-information capturing local thickness of both a shape and its complement space;  
\item a portfolio of applications that leverage medial fields for efficient rendering, generation of collision proxies, and deferred rendering;
\item medial field computation as a constraint optimization problem that stores the field function within the parameters of a deep neural network for compact storage, and O(1) query access time.
\end{itemize}

%% file: sec/2_related.tex
\section{Related works}
Our technique is tightly related to implicit neural representations (\Section{representations}), as well as their efficient rendering~(\Section{rendering}).
Medial fields are also heavily inspired by medial axis techniques, which we also briefly review~(\Section{medial}). 

\subsection{Neural implicit representations}
\label{sec:representations}
Inspired by classic work on implicit representation of geometry~\cite{bloomenthal1997introduction,blinn1982generalization}, a recent class of models has leveraged fully connected neural networks as continuous, memory-efficient representations of 3D scene geometry, by parameterizing either a distance field~\cite{park2019deepsdf,michalkiewicz2019implicit,atzmon2019sal,gropp2020implicit,sitzmann2019srns,jiang2020local,peng2020convolutional,chibane2020ndf} or the occupancy function~\cite{mescheder2019occupancy,chen2019learning} of the 3D geometry.
Learning priors over spaces of neural implicit representations enables reconstruction from partial observations~\cite{park2019deepsdf,mescheder2019occupancy,chen2019learning,sitzmann2020metasdf}.
To learn distance fields when ground-truth distance values are unavailable, we may solve the Eikonal equation~\cite{atzmon2019sal,sitzmann2019siren,gropp2020implicit}, or leverage the property that points may be projected onto the surface by stepping in the direction of the gradient~\cite{ma2020neural}.
One may leverage hybrid implicit-explicit representations, by locally conditioning a neural implicit representation on features stored in a discrete data structure such as a voxel grids~\cite{jiang2020local,peng2020convolutional,chabra2020deep}, octrees~\cite{takikawa2021neural}, or local Gaussians~\cite{genova2019learning}.
However, such hybrid implicit-explicit representations lose the compactness of monolithic representations, complicating the learning of shape spaces.
In this work, we propose to parameterize 3D shape via the medial field, which parameterizes the local thickness of a 3D shape, and gives O(1) access to a number of quantities immediately useful for downstream tasks.

\subsection{Rendering implicits}
\label{sec:rendering}
Rendering of implicit shape representations relies on the discovery of the first level set that a camera ray intersects.
For distance fields, sphere tracing~\cite{hart1996sphere} enables fast root-finding. 
This algorithm known pathological cases, which have been addressed with heuristics~\cite{balint2018accelerating,korndorfer2014enhanced}, as well as coarse-to-fine schemes~\cite{liu2020dist}.
Ours is complementary to these approaches, and we focus our comparison on the core algorithm.
Additionally, sphere tracing can be generalized to the rendering of deformed implicit representations~\cite{seyb2019non}.

For neural implicit shape representations, differentiable renderers have been proposed to learn implicit representations of geometry given only 2D observations of 3D scenes~\cite{sitzmann2019srns,liu2020dist,Niemeyer2020DVR,yariv2020multiview}.
Alternatively, one may parameterize density and radiance of a 3D scene, enabling volumetric rendering~\cite{mildenhall2020nerf}, or combine volumetric and ray-marching based approaches~\cite{oechsle2021unisurf}.
As rendering of neural implicit representations requires hundreds of evaluations of the distance field per ray, hybrid explicit-implicit representations have been proposed to provide significant speedups~\cite{takikawa2021neural}.
As we will demonstrate, the proposed Deep Medial Fields allow fast rendering as they require significantly fewer network evaluations per ray, without relying on a hybrid implicit-explicit representation.

\subsection{Medial axis transform (MAT)}
\label{sec:medial}
The medial axis transform provides a \textit{dual} representation of solid geometry as a collection of spheres.
Computing the medial axis is a challenging problem, due to both the \textit{instability} of this representation with respect to noise~\cite[Fig.3]{lsmat}, and the lack of techniques to compute the medial surfaces given a polygonal mesh~\cite{survey}.
Nonetheless, spherical representations ``inspired'' by the MAT have found widespread use in applications, including shape approximation for static~\cite{thiery2013sphere} and dynamic~\cite{thiery2016animated} geometry, efficient closest point computation~\cite{hmodel}, and volumetric physics simulation~\cite{viper}.
There has been work on constructing the MAT with neural networks~\cite{hu2019mat}, but to the best of our knowledge, ours is the first work to encode medial information in an implicit neural representation.

%% file: sec/3_method.tex
\section{Method}
We start by reviewing the basics of implicit and medial representations (\Section{definitions}), and then introduce the analytical concept of medial fields (\Section{medial_field}).
We then propose a variational formulation of medial fields which allows us to formalize it without requiring us to explicitly compute the medial axis (\Section{variational}), as well as a way of implementing this approach with neural networks (\Section{deep_medial_fields}).

\subsection{Background}
\label{sec:definitions}
Let us consider a (solid) shape $\Shape$ in $d$-dimensional space as partitioning all points $x \in \Real^d$ as belonging to either its interior $\Shape^-$, exterior $\Shape^+$, or (boundary) surface $\Boundary$.
The signed and unsigned distance fields~(respectively SDF and UDF) implicitly represent a shape as:
\begin{align}
\SDF(x) = \begin{cases} 
  +\USDF(x) & x \in \Shape^+ \\
  -\USDF(x) & x \in \Shape^- \\
  0 & x \in \Boundary
\end{cases}
\;,
\quad
\USDF(x) = \min_{y \in \Boundary} ||x - y||
\;,
\end{align}
where the term implicit refers to the fact that the shape boundary is indirectly defined as the zero-crossing of the field.

\input{fig/maxis_def}

\paragraph{The medial axis}
The Medial Axis Transform (MAT) of a shape is the set of ``maximally inscribed spheres''. 
The \emph{medial axis} $\MAxis$ can then be defined as the set of all centers of these spheres:
\begin{align}
    \MAxis = \{x \in \Real^d :\, \forall \delta \in \Real^d \setminus \mathbf{0}, \USDF(x) + ||\delta|| > \USDF(x + \delta) \}
    \;.
    \label{eq:medial}
\end{align}
As illustrated in~\Figure{maxis_def}, note that~$\USDF(x)$ is equal to the radius of an \textit{empty} sphere centered at $x$ and tangent to~$\Boundary$.
Any tangent sphere that can not be grown to a tangent sphere with a $||\delta||$-larger radius by moving the center with some offset $\delta$ is ``maximally inscribed'', and therefore a \emph{medial sphere}.
While there exist multiple ways to define the medial axis~\cite{survey}, we choose this definition as it will allow us to formulate the medial axis in terms of a \textit{field} function: the ``medial field''.

\input{fig/mfield_def}

\subsection{Medial field}
\label{sec:medial_field}
For a point $x \in \Shape^-$, we informally define the medial field as the~``local thickness'' of the shape at $x$, and equivalently for the shape's complement space when $x \in \Shape^+$; see~\Figure{mfield_def}~(left).
To formalize this construct, let us start by noting that $\SDF(x)$ allows us to project a point $x$ onto the closest point of the shape surface $\Boundary$ as:
\begin{align}
    \Project{\Shape}(x) = x - \nabla \SDF(x) \SDF(x)
    \;.
\end{align}
Both $x$ and $\Project{\Shape}(x)$ lie on a line segment $\Line(x)$, known as the ``medial spoke''~\cite{siddiqi2008medial}, which begins at $\Project{\Shape}(x)$ and ends at a point on the medial axis; see \Figure{mfield_def}~(right).
We call this point $\Project{\MAxis}(x)$,
and use it to define the medial field as the scalar function:
\begin{align}
    \MField(x) = |\SDF(\Project{\MAxis}(x))|
    \;.
    \label{eq:mf_definition}
\end{align}
In other words, the medial field $\MField(x)$ is the radius of the medial sphere centered at $\Project{\MAxis}(x)$.
Equivalently, the medial field is the length of the medial spoke $\Line(x)$.
$\MField(x)$ is well defined everywhere \textit{except} at $x \in \Boundary$ where we could have a value discontinuity~--~the medial field for interior/exterior might not match.
While above we employ $\Project{\MAxis}$ to define the medial field, the opposite is also possible:
\begin{align}
\Project{\MAxis}(x) = x + \nabla |\SDF(x)| (\MField(x) - |\SDF(x)|)
\;,
\label{eq:medialproj}
\end{align}
but note that $\Project{\MAxis}(x)$ \textbf{is not} the closest-point projection of $x$ onto $\MAxis$, but rather the intersection of the medial spoke with $\MAxis$.

\subsection{Variational medial fields}
\label{sec:variational}
To use medial fields in an application, one must first compute it.
With the definitions above, given $\SDF(x)$ and $\nabla \SDF(x)$, the medial field could be computed explicitly by querying the medial radius at the \textit{intersection} of the medial spoke and the medial axis.
Unfortunately, computing the medial axis, especially in $\Real^3$, is a challenging and open problem~\cite{survey}.
Rather than defining the medial field constructively as in \eq{mf_definition}, we define the medial field in a variational way, so to never require any knowledge about the geometry of the medial axis $\MAxis$.
More formally, we define the medial field as the function that satisfies the following set of \textit{necessary and sufficient} constraints (\Appendix{proof}):
\begin{align}
\forall x &\in \Real^d \setminus \Boundary,
&\MField^*(x) \geq |\SDF(x)|
\;,
\label{eq:maximality_constraint}
\\
\forall x &\in \Real^d \setminus \Boundary,
&\MField^*(x) = |\SDF(\Project{\MAxis^*}(x))|
\;,
\label{eq:inscription_constraint}
\\
\forall x &\in \Real^d \setminus (\Boundary \cup \MAxis),
&\nabla \MField^*(x) \cdot \nabla \SDF(x) = 0
\;.
\label{eq:orthogonality_constraint}
\end{align}

\subsection{Deep medial fields}
\label{sec:deep_medial_fields}
Following the recent success in \textit{compactly} storing 3D field functions within neural networks~(e.g.~occupancy~\cite{chen2019learning,mescheder2019occupancy}, signed distance fields~\cite{park2019deepsdf,atzmon2019sal}, and radiance~\cite{mildenhall2020nerf,rebain2021derf}), we propose to store the medial field within the parameters $\theta$ of a deep neural network $M_\theta(x)$.
While it is in theory possible to to store the medial field values in a grid, for a 3D shape this requires prohibitively large $O(N^3)$ memory where $N$ is linear resolution.
To store the medial field within the network parameters $\theta$, we enforce the constraints of \Section{variational} stochastically over random points sampled over $\Real^d$ via the losses:
\begin{align}
\Loss{maximal} &= \IE_{x\sim\Real^d}\left[ max(|\SDF(x)| - \MField_\theta(x), 0)^2 \right]
\;,
\\
\Loss{inscribed} &= \IE_{x\sim\Real^d}\left[ (|\SDF(\Project{\MAxis}(x))| - \MField_\theta(x))^2 \right]
\;,
\\
\Loss{orthogonal} &= \IE_{x\sim\Real^d}\left[ (\nabla \MField_\theta(x) \cdot \nabla \SDF(x))^2 \right]
\;.
\end{align}
The architecture of $M_\theta(x)$ is based on simple multi-layer perceptrons (MLPs) that is detailed and analyzed in~\Section{implementation}, while the derivatives $\nabla$ are computed by auto-differentiation in JAX~\cite{jax}.

%% file: fig/maxis_def.tex
\begin{figure}[t]
\begin{center}
\fbox{\includegraphics[height=.4\linewidth]{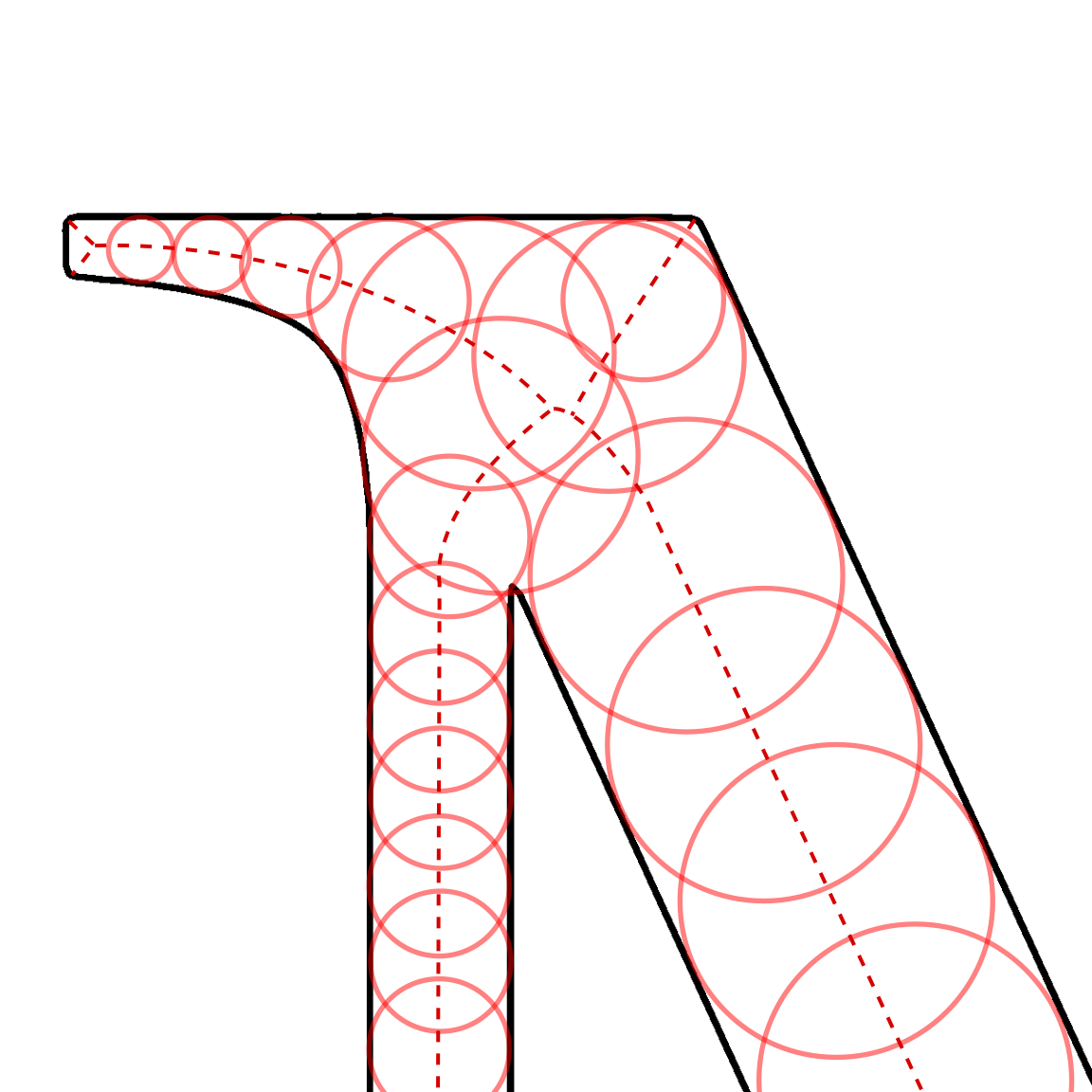}}
\hspace{.2in}
\fbox{\includegraphics[height=.4\linewidth]{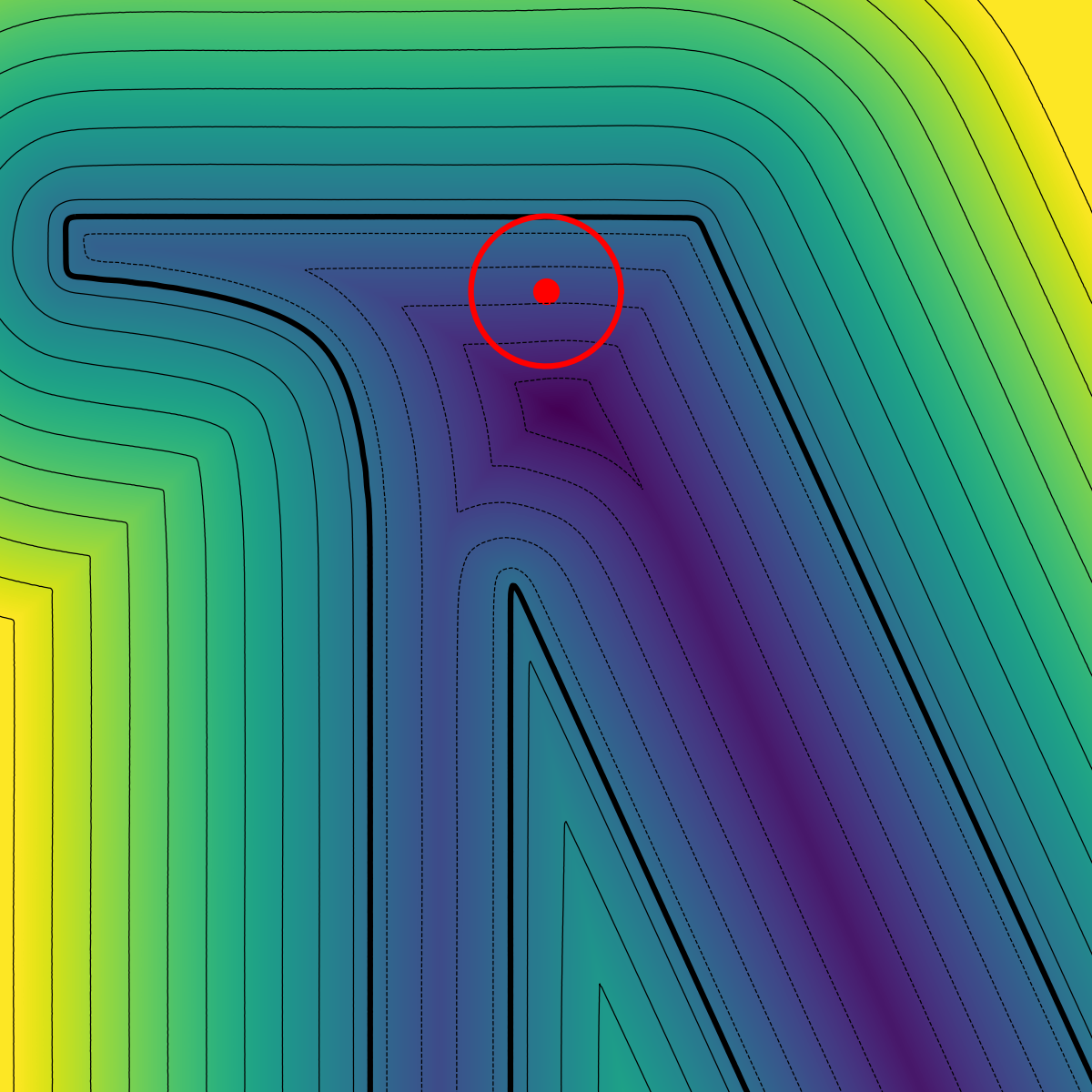}}
\includegraphics[height=.4\linewidth]{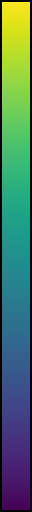}
\end{center}
\vspace{-1em}
\caption{
\textbf{Medial axis and the SDF} --
(left) Interpretation of the medial axis as the set of maximally inscribed sphere.
(right) The signed distance function $\SDF(x)$ of the shape, and a tangent sphere.
}
\label{fig:maxis_def}
\end{figure}

%% file: fig/mfield_def.tex
\begin{figure}[t]
\begin{center}
\fbox{\includegraphics[height=.4\linewidth]{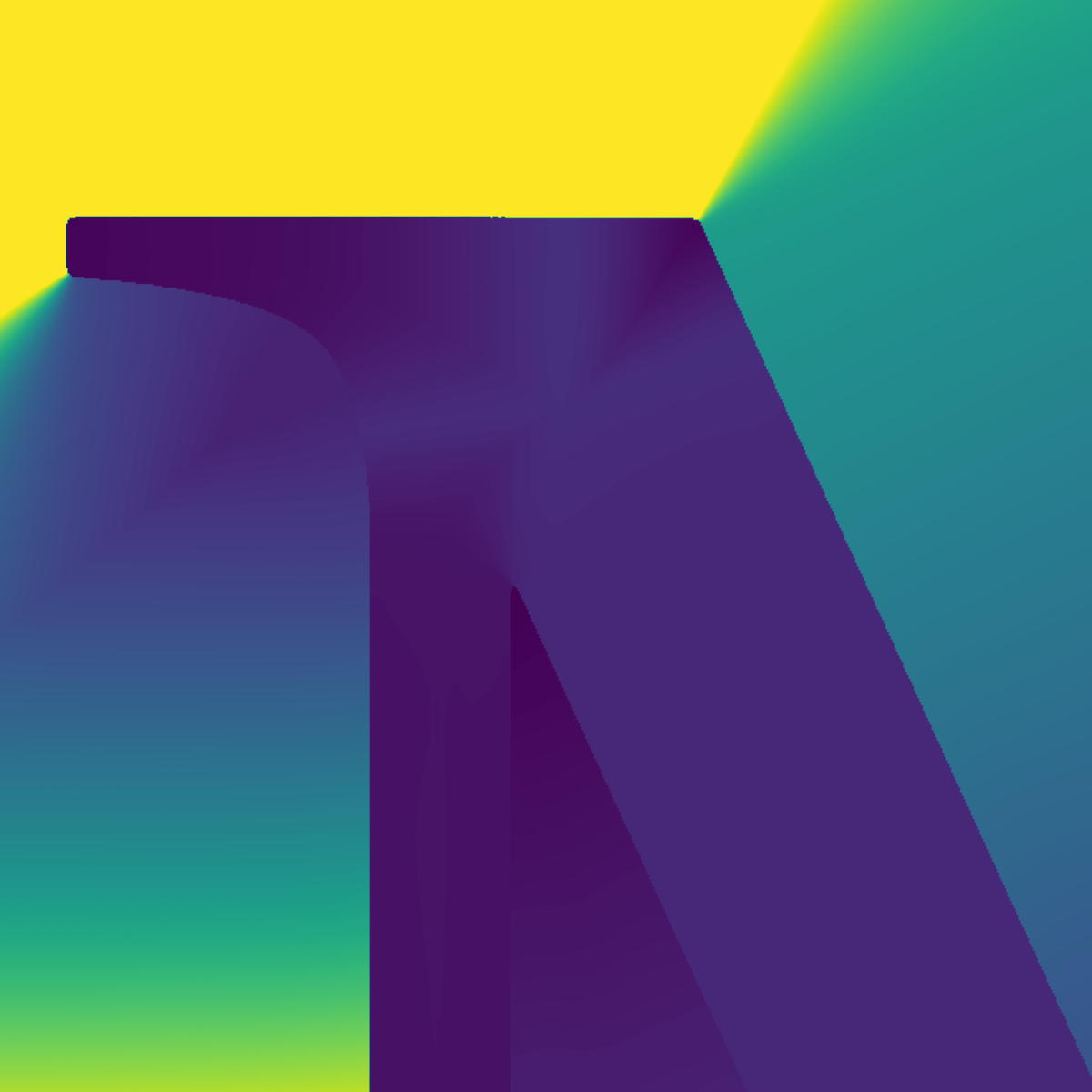}}
\includegraphics[height=.4\linewidth]{fig/colorbar.png}
\hspace{.2in}
\fbox{\includegraphics[height=.4\linewidth]{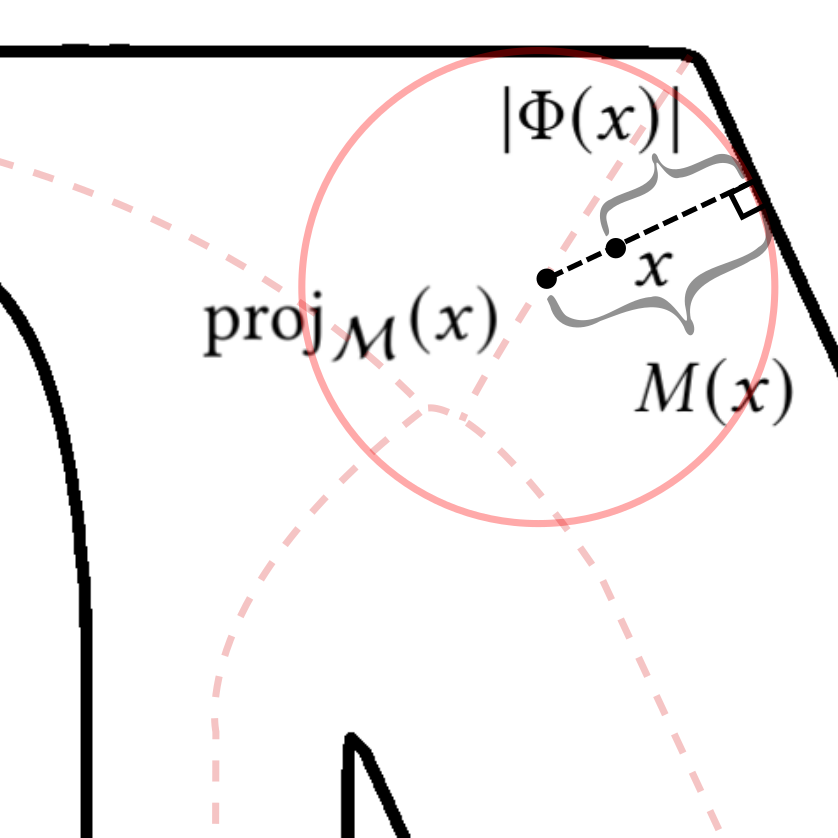}}
\end{center}
\vspace{-1em}
\caption{
\textbf{Medial field} --
(left) Visualization of the scalar medial field;
(right) The notation we use to define the medial field; note that for $x \in \MAxis$ the medial field and the unsigned distance function satisfy~$\MField(x) = |\SDF(x)|$.
}
\label{fig:mfield_def}
\end{figure}

%% file: sec/4_applications.tex
\input{fig/spheretracing}
\input{fig/medialtracing}

\section{Applications}
\label{sec:applications}

\subsection{Rendering implicits}
An algorithm frequently used in conjunction with SDF representations is \emph{sphere tracing}~\cite{hart1996sphere}.
The SDF representations ensure that for any point $x$ in space, a sphere centered at $x$ with radius $|\SDF(x)|$ does not cross the surface.
As such, the SDF value can be used to bound the step size of a ray-marching algorithm in a way that guarantees that overstepping will not occur.
In many cases, such as the case of a ray pointed directly orthogonal to a flat surface, sphere tracing will converge in one or very few iterations, making it an attractive option for rendering which is frequently used in applications that handle implicit surfaces~\cite{seyb2019non,sitzmann2019srns,yariv2020multiview,oechsle2021unisurf,takikawa2021neural}.
However, there are pathological cases in which sphere tracing may take arbitrarily many iterations for a ray to converge.
An example is shown in \Figure{spheretracing}, in which the step sizes become very small as the ray passes close to a surface which it does not intersect.

\paragraph{Medial sphere tracing}
To address this shortcoming of the sphere tracing algorithm, we propose a modification that takes advantage of the additional information encoded in the medial field.
Standard sphere tracing does not take advantage of the \textit{direction} of a ray to avoid cases where the step size should not approach zero near the surface, as all queried spheres are centered on the query point.
Medial spheres, on the other hand, do not suffer this limitation.
As a query point $x$ approaches a smooth surface, the projected medial sphere centered at $\Project{\MAxis}(x)$ will approach a \emph{non-zero} radius that depends on the local thickness of the shape's complement in the region containing $x$.
By using the medial sphere to advance the ray-marching process, many pathological situations in sphere tracing can be avoided; see \Figure{medialtracing}.

\input{fig/st_performance}
\paragraph{Evaluation}
To evaluate the efficacy of our proposed improvement to the sphere tracing algorithm, we perform experiments measuring the impact on the number of iterations taken for rays to converge to the surface.
For this purpose, we train our model on a set of 3D shapes and render each from a random distribution of camera poses.
As shown in \Figure{st_performance}, we find that the medial sphere tracing algorithm yields significantly better worst-case performance compared to the naive algorithm for both iteration counts and render times.
We also analyze the distribution of iterations required to reach convergence for each ray and plot the resulting histograms in \Figure{spheretracing} and \Figure{medialtracing}.
The histogram for the modified algorithm shows a much higher fraction of rays which converge in fewer than 6 iterations, as well as a much smaller tail of rays which take many iterations to converge.
Notice that the quality of the traced image is visually equivalent to the one produced by classical SDF; see~\Figure{st_performance}.

\input{fig/physics}

\subsection{Physics proxy generation}
Collision detection, and the computation of the corresponding collision response lie at the foundation of real-time physics simulation.
In a modern simulation package such as NVIDIA PhysX~\cite{macklin2014unified,macklin2020local}, we find that a set of spheres is used as a compact approximation of geometry for collision detection.
At the same time, storing the gradient of the SDF at the sphere origin location provides the necessary meta-information to compute the collision response vector.
Representing all solid objects within a scene as collection of spheres (i.e.~particles) is convenient, as the complexity of collision response source code increases \textit{quadratically} in the number of physical proxies that are used.
While \citeN{macklin2014unified} employs a \textit{uniform} grid of spheres to approximate the geometry, we reckon this might not be optimal (e.g. to represent the geometry of a simple spherical object with precision $\varepsilon$ we would still need $O(\varepsilon^{-d})$ primitives.
The fundamental question we ask is ``\textit{How can we approximate an object with fewer, larger spheres, rather than many smaller ones?}''
Towards this objective, we exploit the medial axis via its ``maximally inscribed spheres'' interpretation, together with the fact that we store their properties implicitly within our medial field.

\paragraph{Furthest sphere sampling}
To create physics proxies we propose an algorithm that could be understood as a generalization of the furthest point sampling~\cite{fps}, but for spherical data.
We start by placing $N$ points randomly throughout the volume and for each of these points finding a corresponding point on the medial axis, in constant time using the medial projection operation in \eq{medialproj}.
We then draw $M {<} N$ sample points $x_n^*$ through an iterative sampling:
\begin{equation}
x_n^* = \argmax_{x_n} \min_{x_m} \underbrace{{||x_n - x_m||}/{(r_n + r_m + \epsilon})}_\text{normalized separation}
\;.
\label{eq:fss}
\end{equation}
Here, $\epsilon$ prevents the selection of small, geometrically irrelevant spheres, and the normalization provides \textit{scale invariance}: with $\epsilon=0$, note that two spheres of the same size that touch tangentially have a normalized separation equal to one, and this remains the case as we double their size.
This algorithm results in a set of spheres that greedily minimise overlap with each other, and by virtue of them being medial spheres, are locally maximal, and represent as much volume as possible.

\paragraph{Evaluation}
We evaluate the effectiveness of a number of approximations by comparing the representation accuracy (i.e.~surface MAE) vs. memory consumption (\# floats).
We construct spherical proxies for a set of 3D shapes using \CIRCLE{1} medial spheres, \CIRCLE{2} tangent spheres, \CIRCLE{3} uniform spheres, as well as an SDF discretization~(as a collision detection event can be computed by just checking whether $\SDF(x){<}0$ for a query point $x$).
The tangent spheres are sampled uniformly, and use radii provided by the SDF.
The uniform spheres are sampled on a regular grid, and use a radius bounded by the width of the grid cells.
The SDF grid also uses a regular grid, but represents the surface by interpolating the SDF values at the grid corners.
We find that for a fixed memory budget, the medial spheres computed from the medial field provide the most accurate surface representation; see \Figure{physics}~(bottom).

\input{fig/ao}
\subsection{Ambient occlusion}
Ambient occlusion is used to increase realism in rendering by \textit{modulating} ambient lighting by the fraction of views into the surrounding environment that are occluded by local geometry~\cite{ambientocclusion}.
Computing this value correctly requires integrating over all rays leaving a point on the surface, which can have significant cost for even moderately complicated geometry.
Therefore, it is often approximated by a variety of methods which provide visually similar results at \textit{substantially} lower computational cost~\cite{ssao}.

\paragraph{Screen-space ambient occlusion~(SSAO)}
One popular example of efficient ambient occlusion approximation is \textit{screen-space ambient occlusion}, which computes it using the \textit{fraction} of depth values in the region of the screen surrounding a point that are smaller than the query point's depth~\cite{ssao}.
Note SSAO is a deferred rendering technique, as it requires the depth map to be rasterized first, and as it is based on random sampling it needs a secondary de-noising phase to remove noise.
Further, as its outcome is view-dependent, the visual appearance of SSAO is not necessarily stable with respect to viewpoint variations.

\paragraph{Medial field ambient occlusion~(MFAO)}
We introduce \textit{medial field ambient occlusion} as an efficient alternative to approximate ambient occlusion whenever an O(1) medial field is readily available.
The medial field provides information about ``local shape thickness'' of both the shape, and more importantly for this application, its \textit{complement space} $\Shape^+$.
We employ the medial field to measure the local thickness of the shape complement, and derive from this value a proxy for local occlusion at a surface point $x$:
\begin{align}
MFAO(x) = \min(a \MField(x + \nabla \SDF(x) \epsilon)^p, 1)
\;,
\label{eq:mfao_def}
\end{align}
where $p$ and $a$ are parameters that control the strength of the effect, and $\epsilon$ is a small offset to ensure that the sampled point is in the shape complement $\Shape^+$.
Similarly to SSAO, this method is \textit{local}, in that it does not consider the effect of distant geometry, but behaves as expected in providing modulation of the ambient lighting in \textit{concave areas} of the surface.

\paragraph{Analysis}
We qualitatively compare MFAO to SSAO in \Figure{ao}.
Because both methods aim only to emulate the \textit{appearance} of real ambient occlusion, and make no effort to approximate it in a mathematical sense, there is no meaningful way to perform a quantitative comparison.
Nonetheless, MFAO has some immediate advantages when compared to SSAO:
\CIRCLE{1} as it does not require random sampling in its computation, it also does not requires a smoothing post-processing step to remove noise;
\CIRCLE{2} unlike SSAO, which is by its nature view-dependent, it depends only on the medial field value evaluated near the surface, and is therefore stable and consistent as the viewpoint changes.

%% file: fig/spheretracing.tex
\begin{figure}[t]
\begin{minipage}{\linewidth}
\begin{center}
{
\setlength{\fboxrule}{.5pt}
\fbox{\includegraphics[width=.945\linewidth]{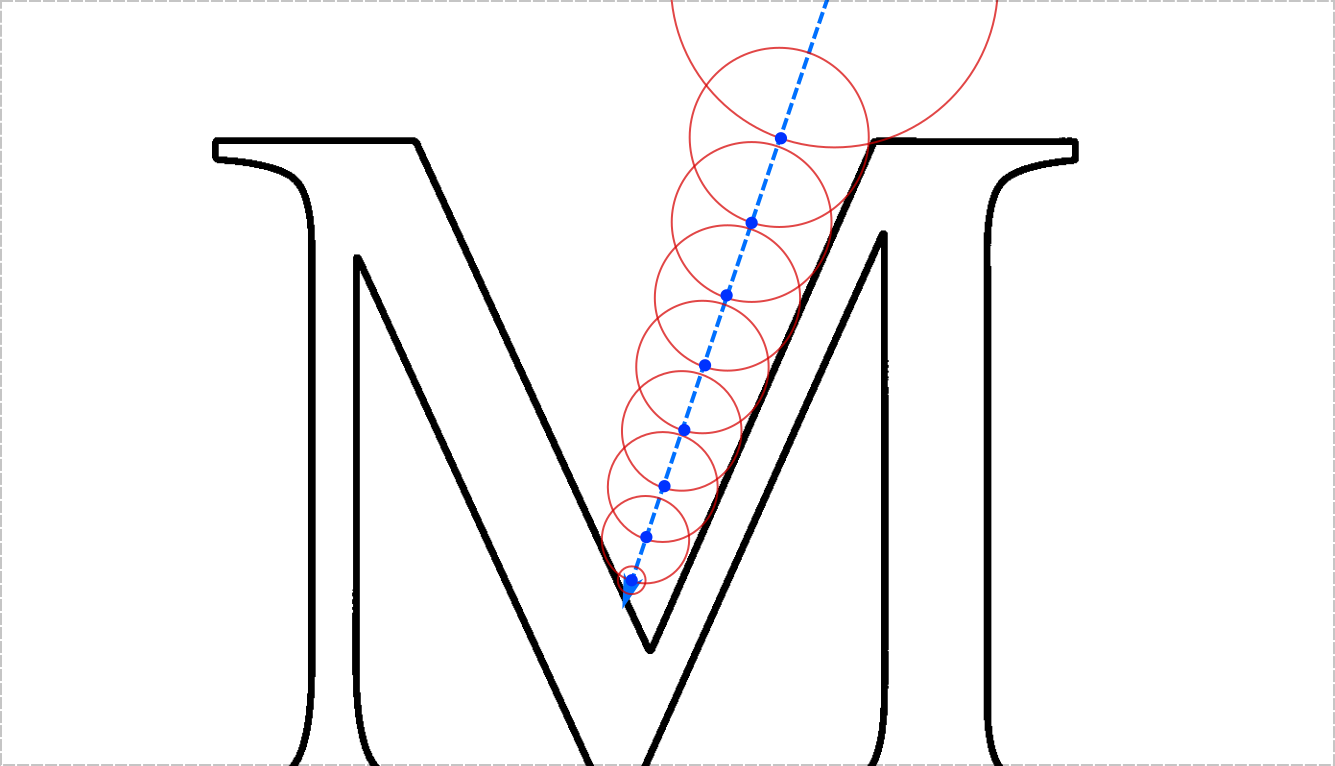}}
\includegraphics[width=.95\linewidth]{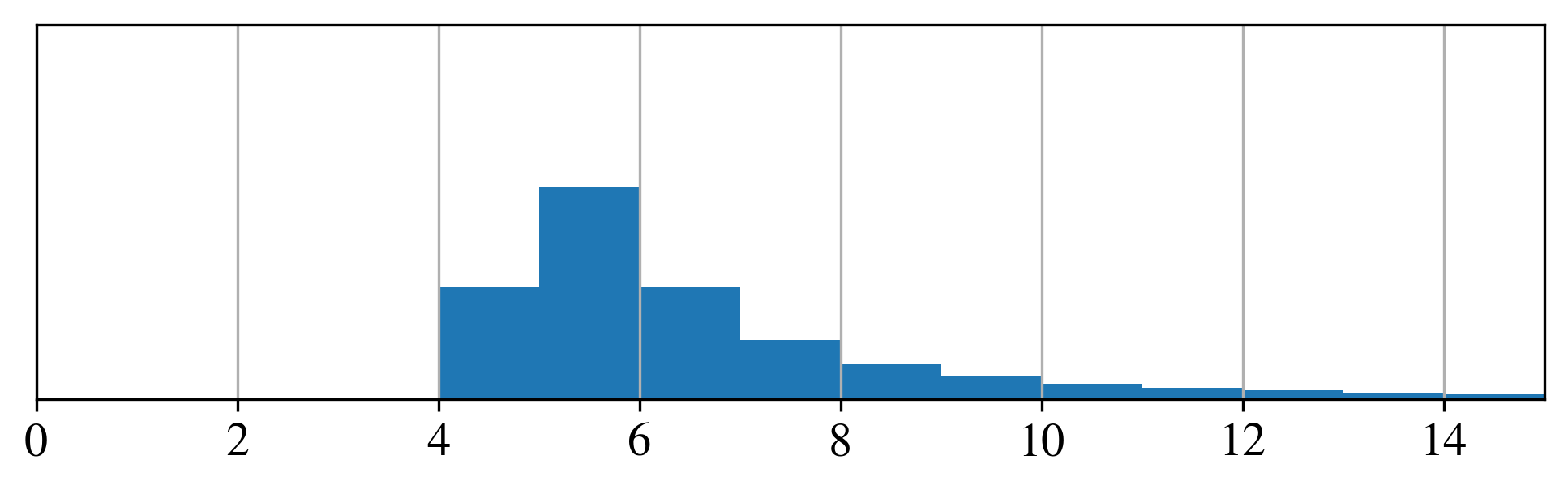}
\\[-.25em]
{
\setlength{\fboxsep}{2pt}
\fbox{%
\begin{minipage}{.925\linewidth}
\SetAlgoNoLine
\KwIn{Ray direction $d$ and origin $o$}
\KwOut{Position $x$ of the ray intersection with $\Boundary$}
$x$ = $o$\\
\Repeat{$|\SDF(x)| < \epsilon$}{
    $x \gets x + \SDF(x) d$
}
\end{minipage}
} 
}}
\end{center}
\vspace{-1em}
\caption{
\textbf{Sphere tracing} --
(bottom) The sphere tracing algorithm introduced by~\cite{hart1996sphere} results in a long-tailed distribution of iterations when rendering 3D scenes (middle).
This is caused by pathological configurations where rays graze the surface of objects (top).
Note the histogram is computed across all of our 3D test scenes, and the 2D example is illustrative in purpose.
}
\label{fig:spheretracing}
\end{minipage}
\end{figure}

%% file: fig/medialtracing.tex
\begin{figure}[t]
\begin{center}
{
\setlength{\fboxrule}{.5pt}
\fbox{\includegraphics[width=.945\linewidth]{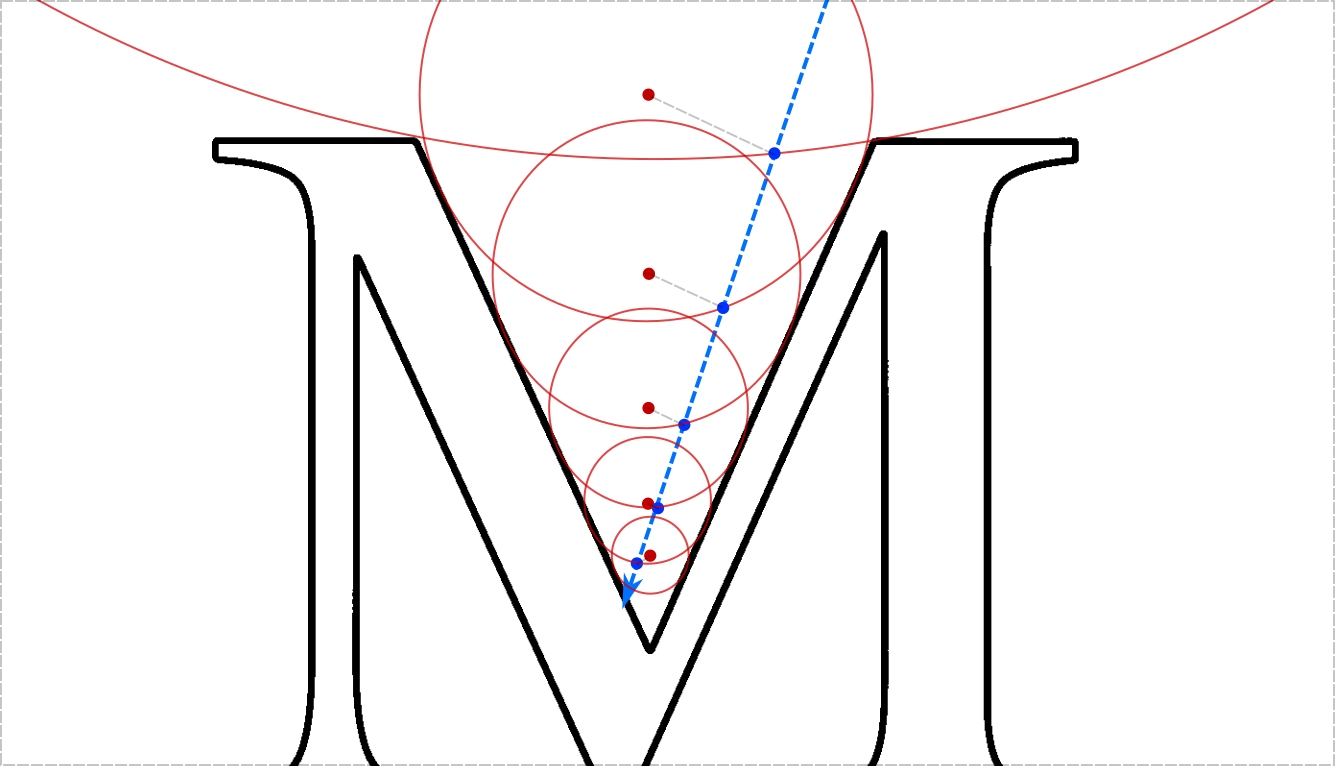}}
\includegraphics[width=.95\linewidth]{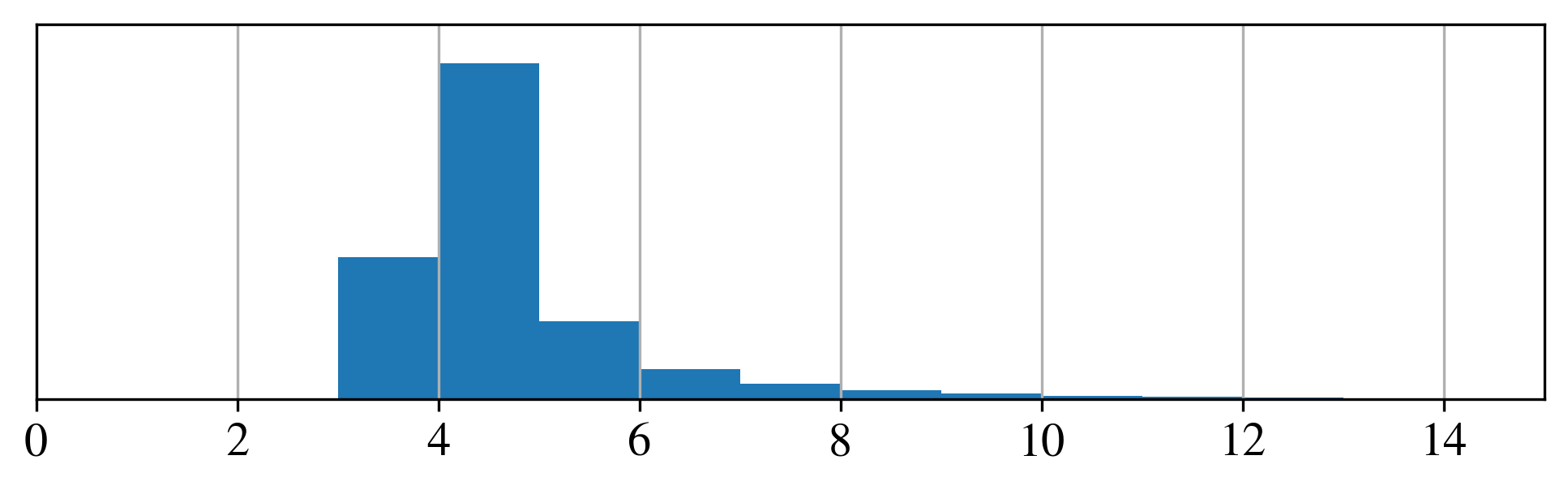}
\\[-.25em]
{\setlength{\fboxsep}{2pt}
\fbox{%
\begin{minipage}{.925\linewidth}
\SetAlgoNoLine
\KwIn{Ray direction $d$ and origin $o$}
\KwOut{Position $x$ of the ray intersection with $\Boundary$}
$x$ = $o$ \\
\Repeat{$|\SDF(x)| < \epsilon$}{
    $\Project{\MAxis}(x)$ = $x + \nabla |\SDF(x)| (\MField(x) - |\SDF(x)|)$ \\
    $\beta$ = $(\Project{\MAxis} - x) \cdot d$ \\
    $\alpha$ = $\sqrt{\beta^2 - (||\Project{\MAxis}(x) - x||_2^2 - \MField(x)^2)}$ \\
    $s$ = $\alpha + \beta$ \\
    $x \gets x + s d$
}
\end{minipage}
}}} 
\end{center}
\vspace{-1em}
\caption{
\textbf{Medial sphere tracing} --
Our algorithm exploits the medial field to use larger spheres to advance tracing along the ray.
As such, the expected number of query iterations is smaller and its distribution has a much shorter tail, leading to faster render time.
Note the histogram is computed across all of our 3D test scenes, and the 2D example is illustrative in purpose.
}
\label{fig:medialtracing}
\end{figure}

%% file: fig/st_performance.tex
\begin{figure}[t]
\begin{center}
\includegraphics[width=.48\linewidth, trim=50 80 40 100, clip]{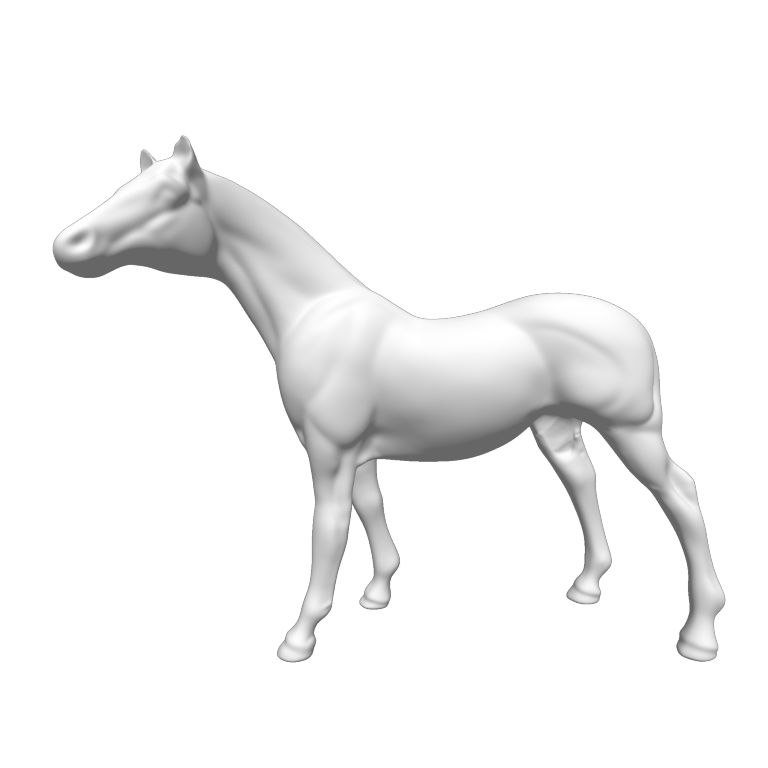}
\includegraphics[width=.48\linewidth, trim=50 80 40 100, clip]{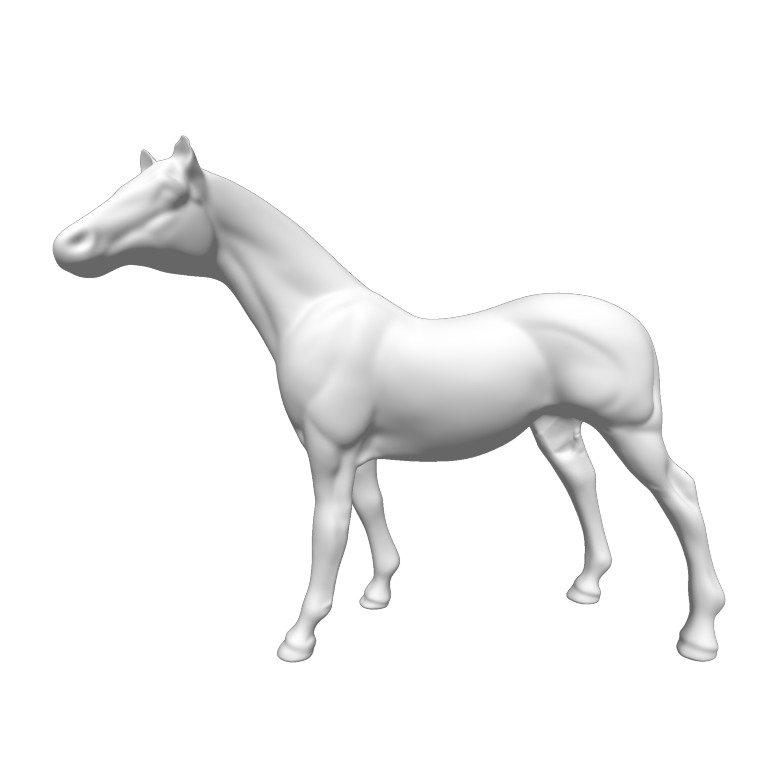}

\vspace{-2em} 
\includegraphics[width=.48\linewidth, trim=50 40 50 40, clip, angle=90, origin=c]{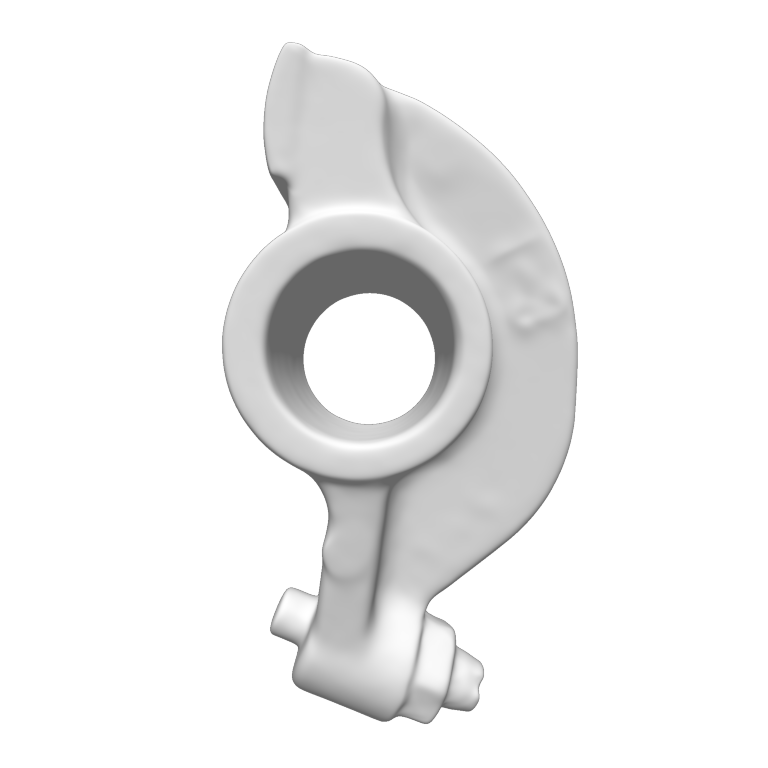}
\includegraphics[width=.48\linewidth, trim=50 40 40 40, clip,  angle=90, origin=c]{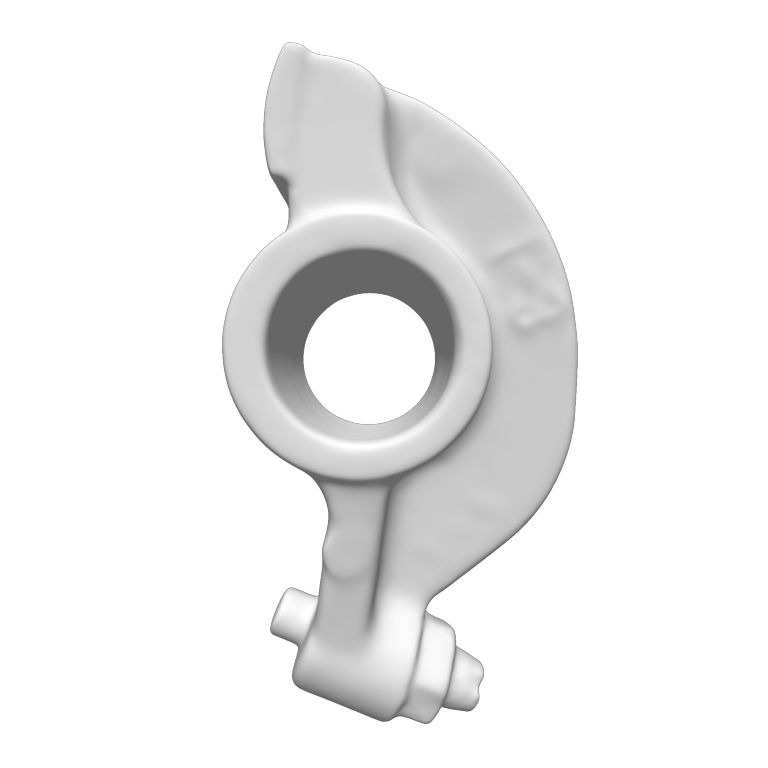}
\vspace{-2em} 

\includegraphics[width=.48\linewidth, trim=50 40 40 100, clip]{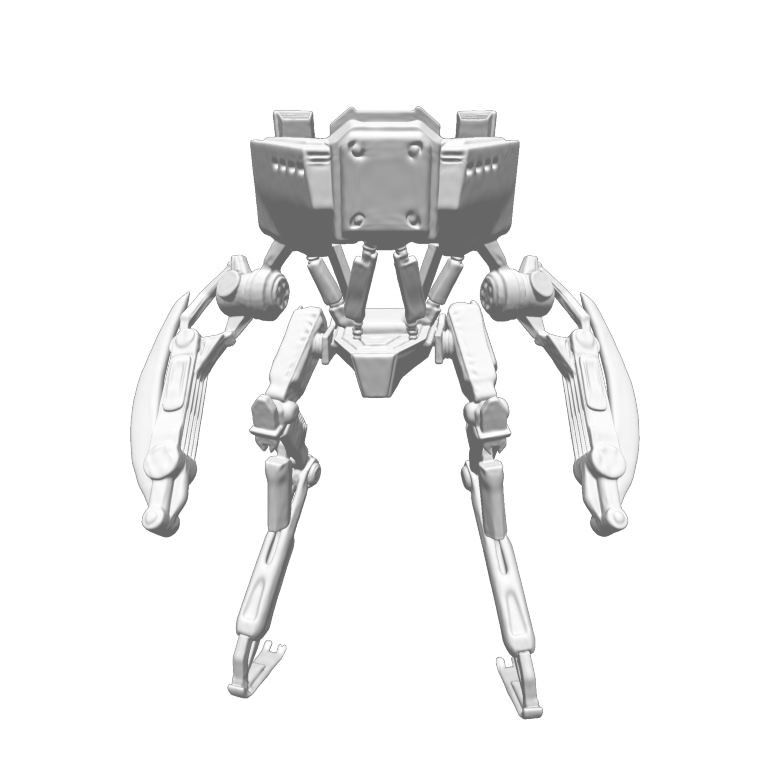}
\includegraphics[width=.48\linewidth, trim=50 40 40 100, clip]{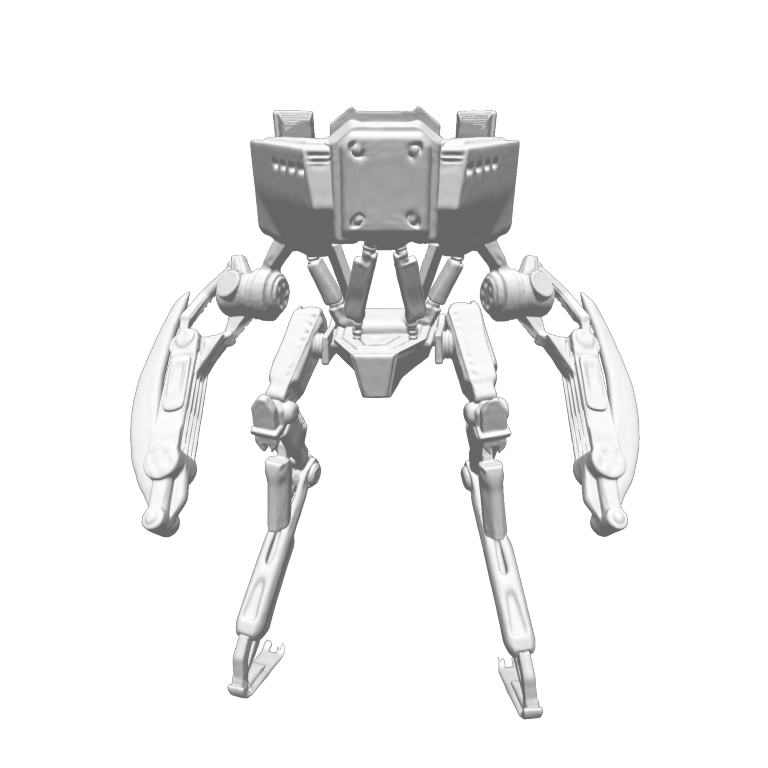}
\end{center}
\begin{center}
\vspace{-1em}
\resizebox{\linewidth}{!}{
\begin{tabular}{r ccc ccc}
\toprule
\multicolumn{1}{c}{\multirow{2}{*}[-2pt]{Scene}}          & \multicolumn{3}{c}{Naive Sphere Tracing} & \multicolumn{3}{c}{Medial Sphere Tracing} \\
           \cmidrule(r){2-4}
           \cmidrule(r){5-7}
     & Mean    & Min    & Max          & Mean    & Min    & Max          \\ 
\midrule
armadillo  & 7.1     & 6.5        & 8.2              & 4.7     & 4.3        & 5.1              \\
bunny      & 7.4     & 6.7        & 8.3              & 4.5     & 4.2        & 4.8              \\
horse      & 6.5     & 6.0        & 7.3              & 4.1     & 3.8        & 4.5              \\
lucy       & 6.1     & 5.6        & 6.9              & 4.1     & 3.8        & 4.5              \\
mecha      & 6.9     & 6.3        & 7.8              & 4.8     & 4.4        & 5.2              \\
rocker-arm & 6.0     & 5.6        & 6.5              & 3.7     & 3.5        & 4.0           \\
\bottomrule
\end{tabular}
} 
\end{center}
\vspace{-1em}
\caption{
\textbf{Rendering implicits} --
Three of the six scenes used for our evaluation rendered by (top-left) naive tracing and (top-right) medial tracing; note the rendered images are visually indistinguishable.
(bottom) Our quantitative analysis reporting statistics about the average number of tracing iterations when rendering a frame for ``naive'' (left) vs. ``medial'' (right) sphere tracing.
}
\label{fig:st_performance}
\end{figure}

%% file: fig/physics.tex
\begin{figure}[t]
\begin{center}
\includegraphics[width=.44\linewidth, trim = 40 0 40 0, clip]{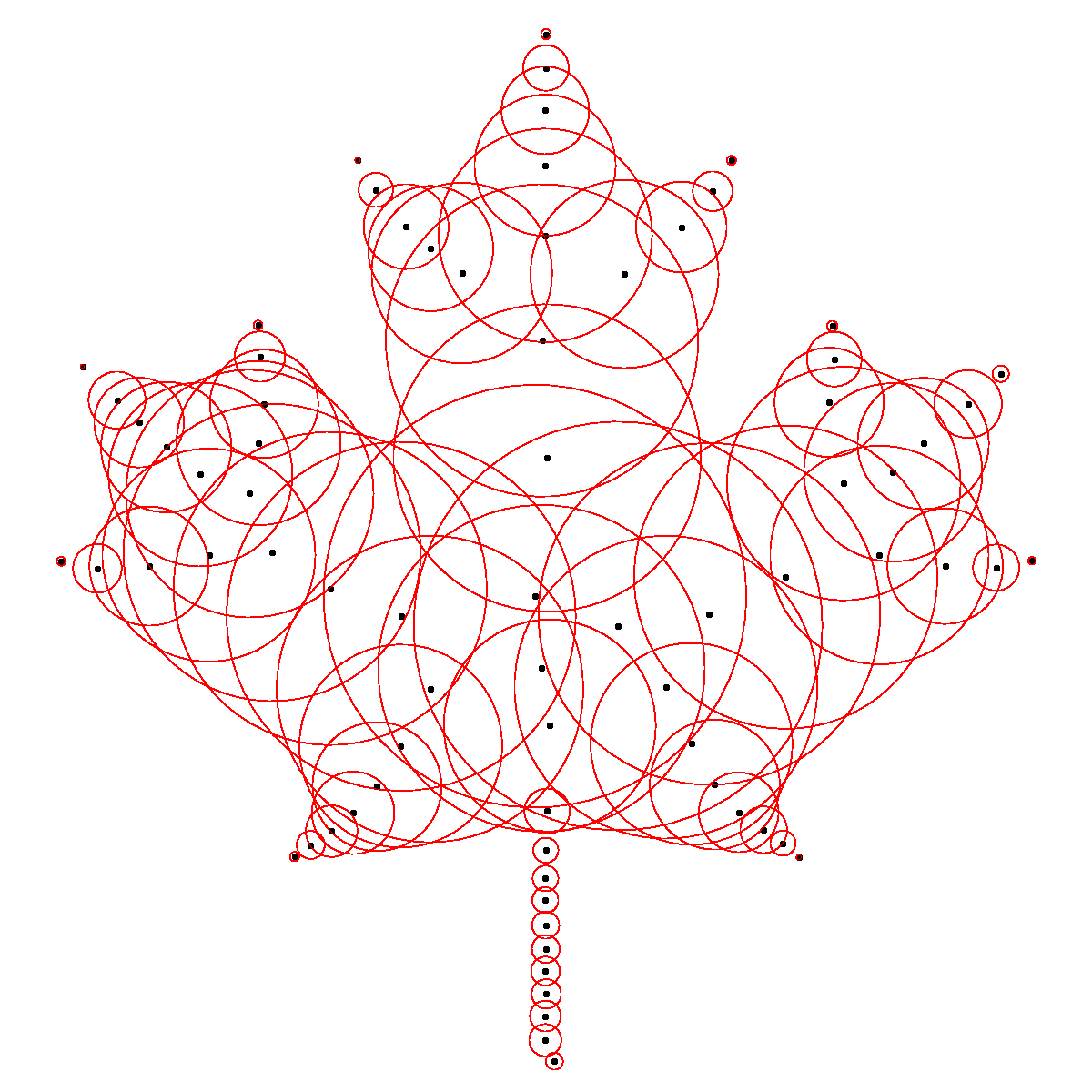}
\hfill
\includegraphics[width=.44\linewidth, trim = 40 0 40 0, clip]{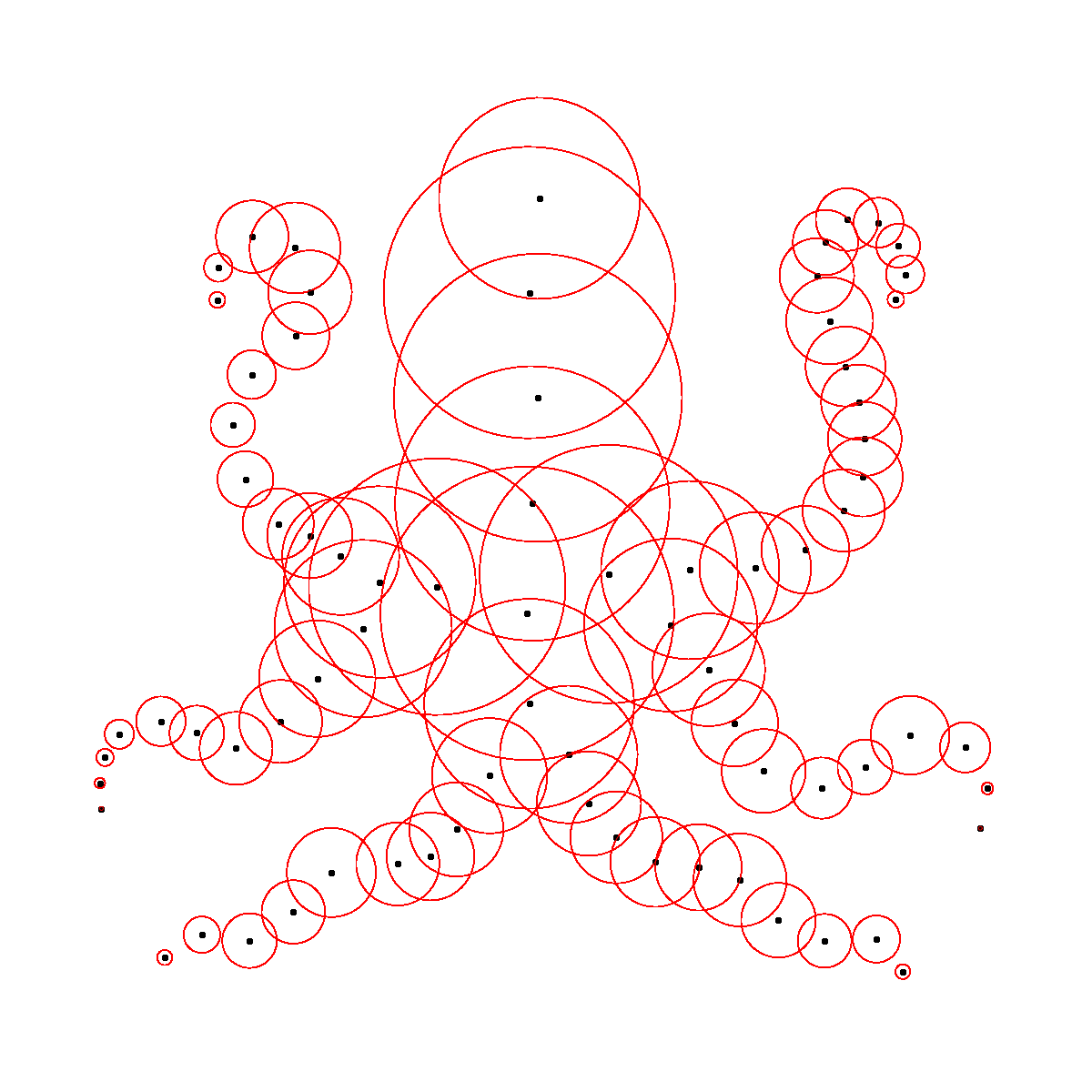}

\vspace{-0.5em}
\includegraphics[width=.47\linewidth, trim=0 80 0 100, clip]{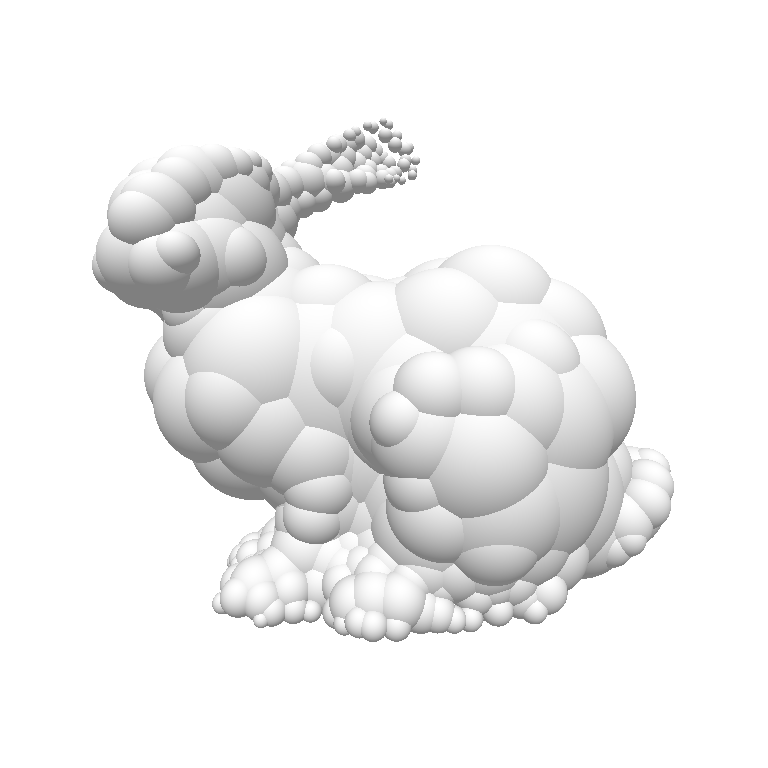}
\includegraphics[width=.49\linewidth, trim=50 80 20 100, clip]{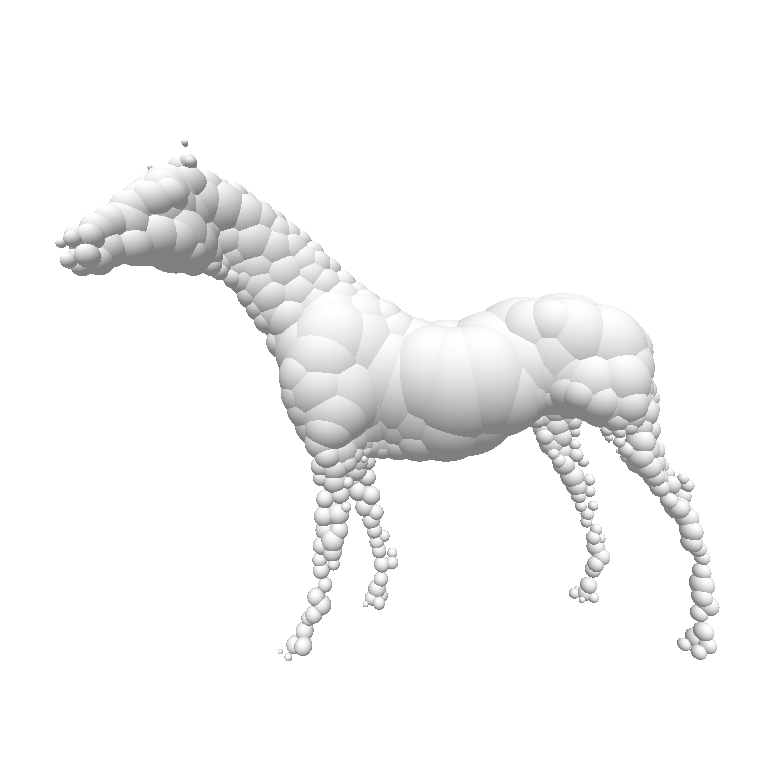}

\includegraphics[height=.4\linewidth, width=\linewidth]{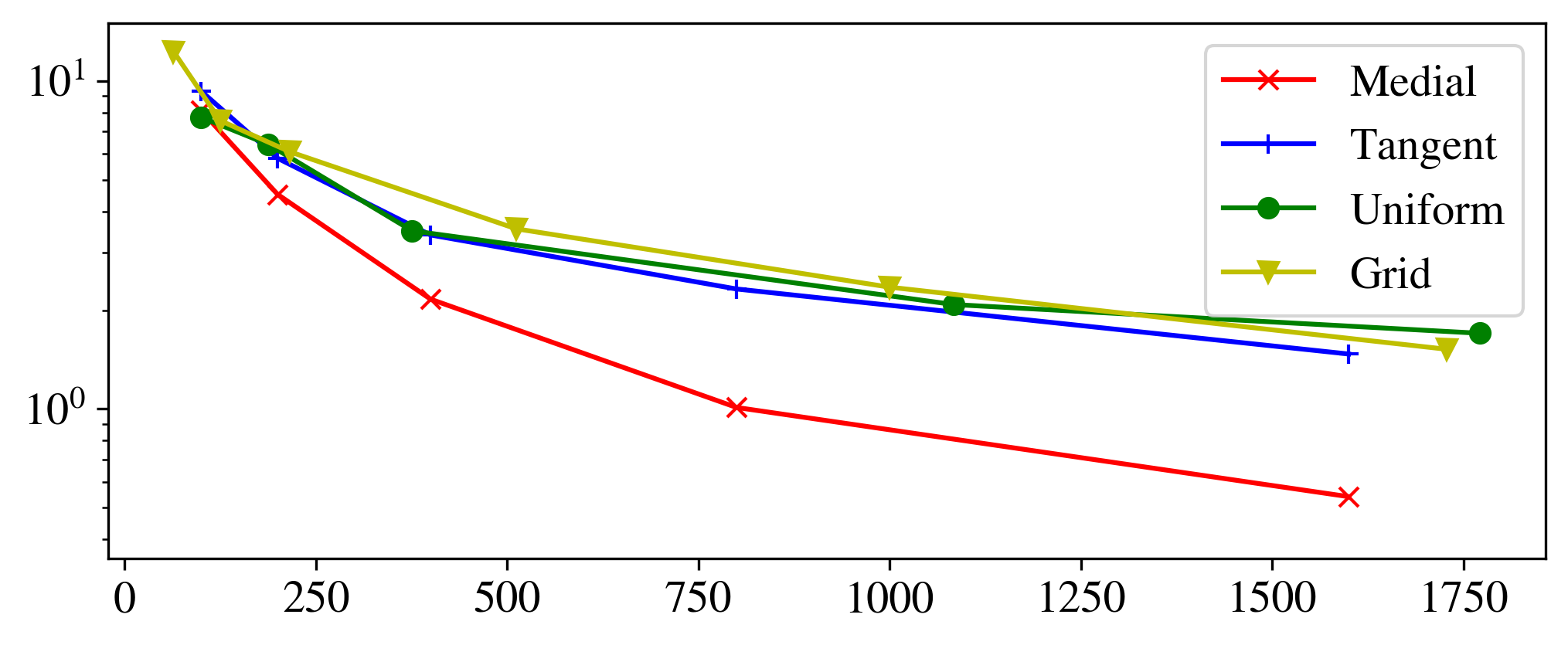}
\end{center}
\vspace{-1em}
\caption{
\textbf{Physics collision proxies} --
Spherical (circular) shape approximation computed by the medial field in 2D~(top) and in 3D~(middle).
We also quantitatively analyze a variety of collision proxies, revealing the corresponding memory/accuracy trade-off in percentage MAE vs memory in \# of floats (bottom).
}
\label{fig:physics}
\end{figure}

%% file: fig/ao.tex
\begin{figure}[t]
\begin{center}
\includegraphics[height=.6\linewidth, trim=100 40 90 100, clip]{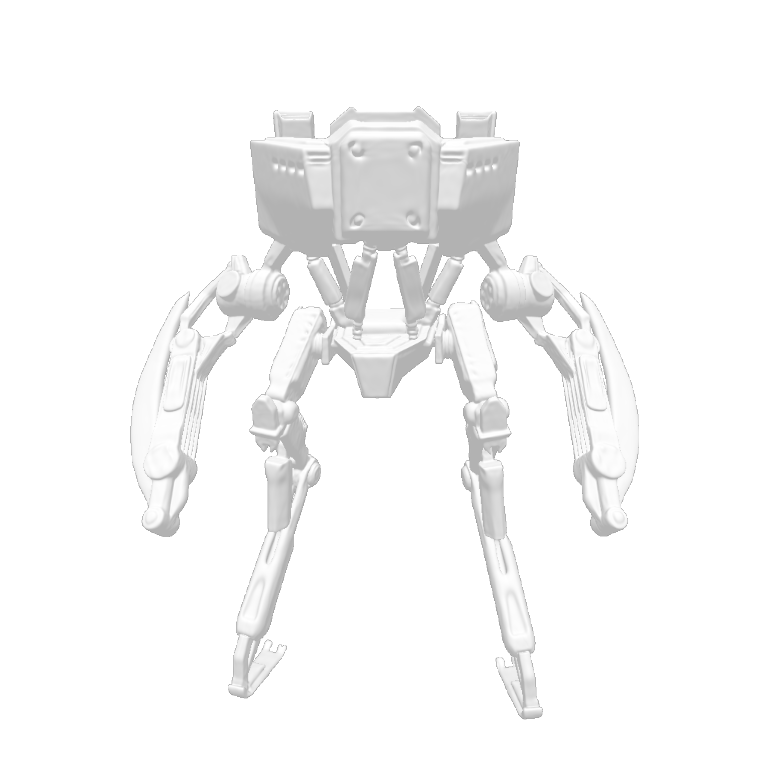}
\includegraphics[height=.6\linewidth, trim=180 50 180 40, clip]{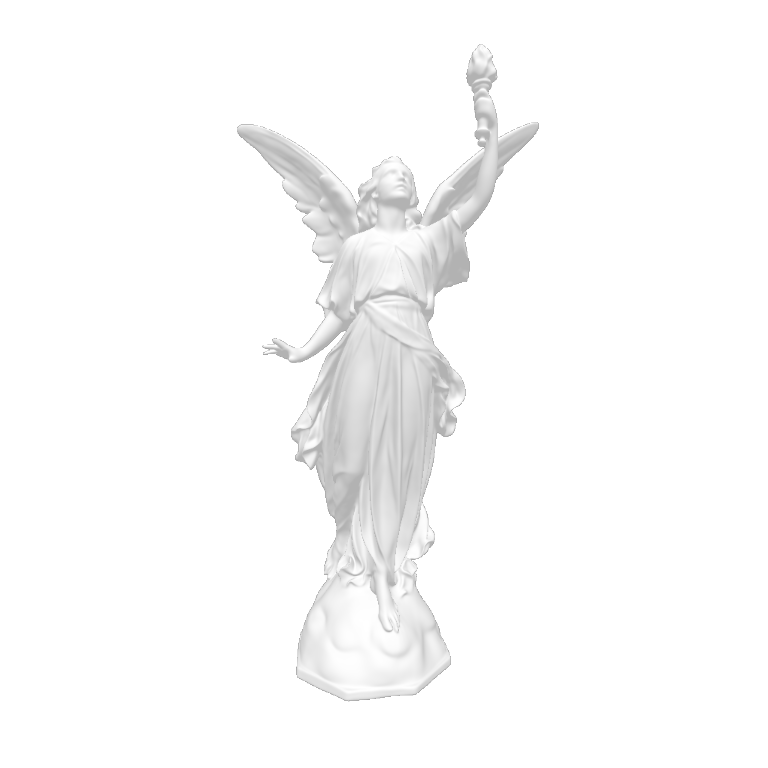}
\includegraphics[height=.6\linewidth, trim=100 40 90 100, clip]{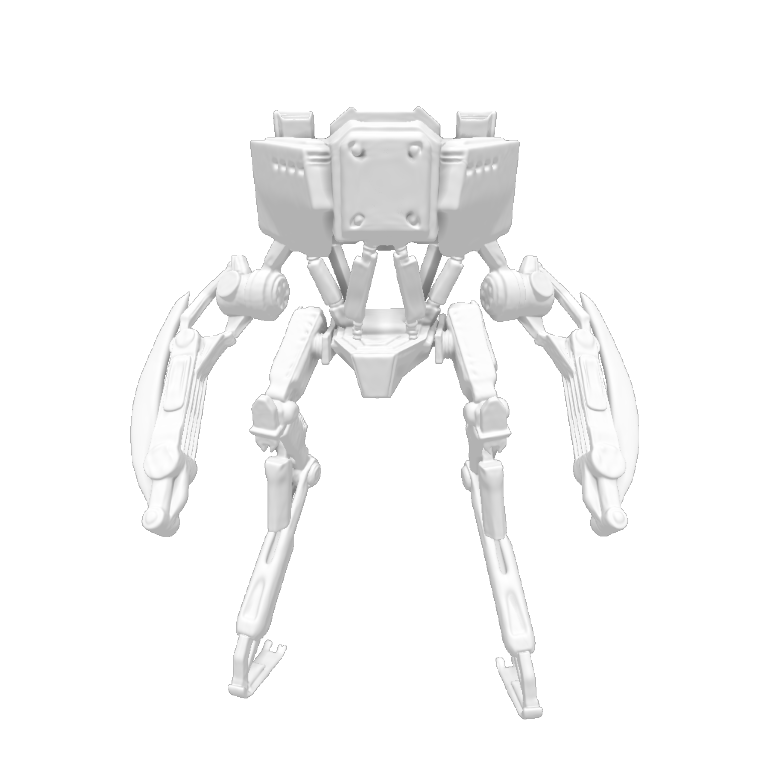}
\includegraphics[height=.6\linewidth, trim=180 50 180 40, clip]{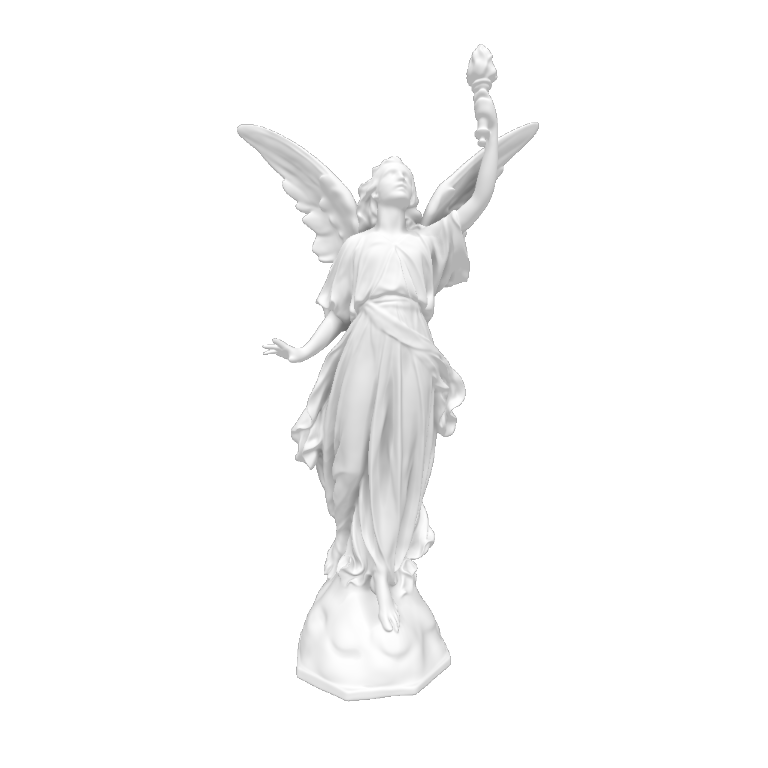}
\includegraphics[height=.6\linewidth, trim=100 40 90 100, clip]{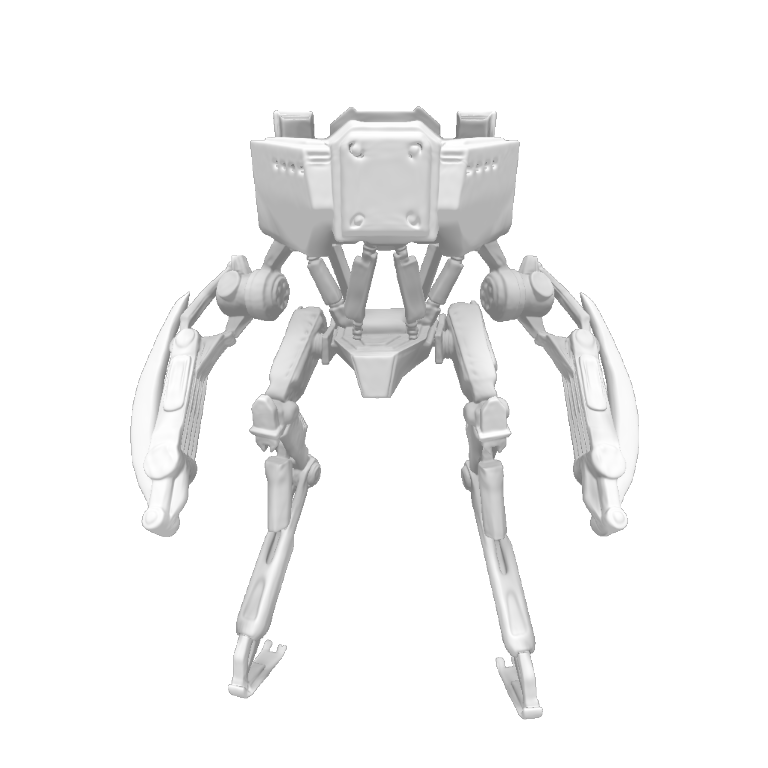}
\includegraphics[height=.6\linewidth, trim=180 50 180 40, clip]{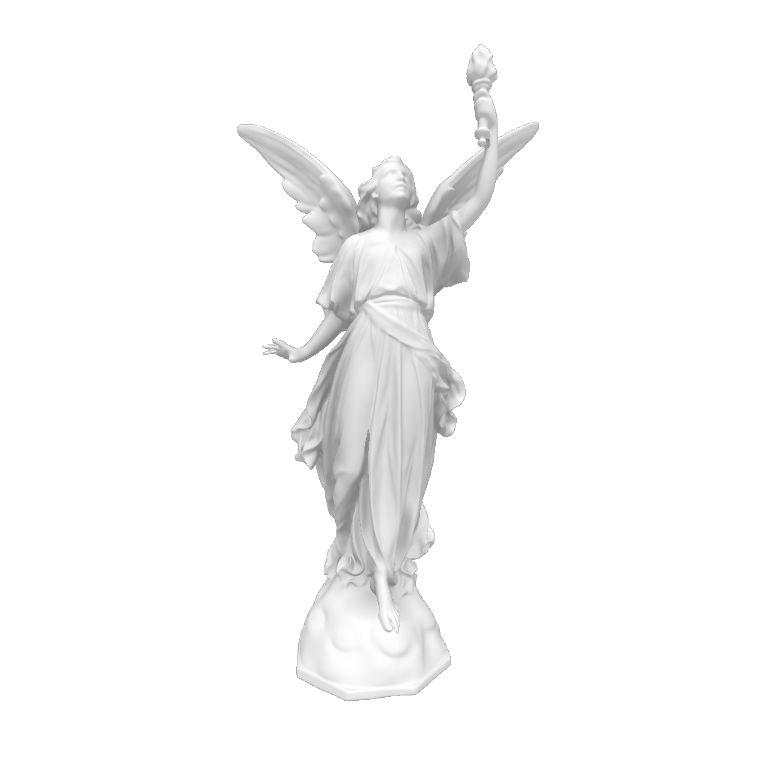}
\end{center}
\vspace{-1em}
\caption{
\textbf{Ambient occlusion} --
Given a smooth-shaded model (top), we visualize classical ``screen space ambient occlusion'' (middle) as well as ``local medial field ambient occlusion'' (bottom).
In comparison to SSAO, MFAO captures occlusion at a range defined by the parameters, without dependence on the viewpoint, and without requiring neither sampling nor filtering.
Examples rendered with MFAO constants $p=0.2$~and~$a=1.5$.
}
\label{fig:ao}
\end{figure}

%% file: sec/5_implementation.tex
\input{fig/architecture}
\section{Implementation}
\label{sec:implementation}

The typical structure for a neural SDF is, as in~\cite{chen2019learning,park2019deepsdf}, a multi-layer perceptron (MLP) network operating on the coordinates of an input point (or an encoded version of it).
Because our constraint formulation for $\MField(x)$ requires a way to compute $\SDF(x)$ and $\nabla \SDF(x)$, we opt to extend a neural SDF architecture to also predict $\MField(x)$ in addition to the typical $\SDF(x)$.
We then have the option to compute $\nabla \SDF(x)$ by back-propagation, or to additionally predict it as an output.

\paragraph{Multi-headed architecture}
We employ a multi-headed network architecture with a shared backbone; see~\Figure{architecture}.
Specifically, we opt for a single MLP to model all values that we wish to regress, which shares a common computation path to build a deep feature that is translated by a dedicated MLP head to either the medial field $\MField(x)$, the SDF $\SDF(x)$, or the gradient of the SDF $\nabla \SDF(x)$.

\paragraph{Modeling a discontinuous field} 
The medial field has a discontinuity at the surface $\Boundary$, which cannot be approximated by ordinary MLPs because they are Lipschitz in the inputs, and learning with gradient-based methods is biased towards small Lipschitz constants~\cite{Bartlett2017SpectrallynormalizedMB,rahaman2019spectral}.
To resolve this issue, we define two components of the medial field: $\MField^-(x)$ and $\MField^+(x)$, which are defined for $x \in \Shape^-$ and $x \in \Shape^+$ respectively.
We predict the values of these fields separately with different heads, and combine them to produce $\MField(x)$:
\begin{align}
    \MField(x) = \begin{cases} 
      \MField^+(x) & \SDF(x) > 0 \\
      \MField^-(x) & \SDF(x) < 0
   \end{cases}
   \;.
\end{align}
This makes it possible to model the discontinuity on $\Boundary$ even though the outputs of all heads are Lipschitz in the input coordinates.

\input{fig/surface_quality}
\subsection{Interference analysis -- \Table{surface_quality}}
To evaluate the effect of adding our medial field-specific losses and/or network components, we perform experiments for a collection of 2D and 3D shapes that measure the accuracy of the resulting surface representations.
Specifically, for each shape we train our DMF model, as well as an equivalent neural SDF model with the medial field-specific losses omitted.
We find that the addition of the medial field-related elements does not substantially affect the reconstruction quality of the network.
In fact, the differences in quality are substantially below the inter-shape variation and might be entirely explained by the stochastic nature of training.

\subsection{Training details}
We adopt a similar neural SDF architecture to~\cite{gropp2020implicit}.
The core of the network is an MLP, operating on the encoded position, with 6 hidden layers using the geometric initialization and activation described in \cite{gropp2020implicit}.
We then transform the resulting feature representation into $\SDF(x)\in\Real^1$, $\MField^+(x)\in\Real^1$, $\MField^-(x)\in\Real^1$, and $\nabla\SDF(x)\in\Real^3$ via dedicated two-layer MLPs (heads), each consisting of two hidden layers with 64 neurons and the same activations as the core, with a final output linear layer.
We differ from \cite{gropp2020implicit} in that we take as input coordinates transformed into random Fourier features~\cite{tancik2020fourier} composed of 64 bands, in addition to the un-transformed coordinates.
To make this compatible with the geometric initialization, we weight the Fourier features using $w_i = \alpha ||f_i||$, where $\alpha$ is a small constant that we set to $10^{-3}$, and $||f_i||$ is the frequency of the $i$-th band, so as to preserve the network's bias towards a spherical topology.
To provide supervision for the SDF predictions we use the surface and normal reconstruction loss:
\begin{align}
    \Loss{surface} &= \IE_{x\sim\Boundary}\left[ \SDF(x)^2 \right]\;, \\
    \Loss{normal} &= \IE_{x\sim\Boundary}\left[ (\nabla \SDF(x) - \nabla \SDF_\textrm{GT}(x))^2 \right]\;,
\end{align}
Where $x\sim\Boundary$ are samples drawn uniformly from the ground-truth surface.
To supervise the behaviour of the SDF in the volume we use the Eikonal loss~\cite{gropp2020implicit}:
\begin{align}
    \Loss{Eikonal} = \IE_{x\sim\Real^d}\left[ (||\nabla \SDF(x)|| - 1)^2 \right]
\;,
\end{align}
Additionally, we employ two regularizers to improve the quality of the predicted surface and prevent the inclusion of Fourier features from interfering with the implicit SDF training:
\begin{align}
    \Loss{minsurface} &= \IE_{x\sim\Real^d}\left[ e^{-100 |\SDF(x)|} \right]\;, \\
    \Loss{curvature} &= \IE_{x\sim\Real^d}\left[ \left\|\frac{\partial}{\partial t} \nabla \SDF(x + t \nabla \SDF(x)) \right\|_1 \right]
\;.
\end{align}
To speed up the execution of medial sphere tracing, which requires the gradient $\nabla \SDF(x)$ at each query point, we predict the value directly with one of the network heads and constrain it to match the analytic value found through back-propagation:
\begin{align}
    \Loss{gradient} = \IE_{x\sim\Real^d}\left[ ||\nabla \SDF(x) - \nabla_\textrm{analytic} \SDF(x)||^2 \right]
\;,
\end{align}

\input{fig/hyperparam}
\paragraph{Training setup}
For all networks we train and evaluate we use the same training setup.
We use the Adam optimizer~\cite{Kingma15} with the default parameters and a batch size of $2^{13}$.
Our complete loss function is a weighted summation of all losses, with empirically determined hyperparameters that provide smooth optimization for all loss terms in \Table{hyperparam}.
As some of our losses depend on point on surfaces $x{\sim}\Boundary$ and others on points in the rendering volume $x{\sim}\Real^d$ we need different sampling strategies to account for this.
To form batches of surface points we sample uniformly over the input surface mesh, using the triangle normals to provide~$\nabla \SDF_\textrm{GT}(x)$.
To form batches of points from $\Real^d$ we augment the surface samples with an offset from an isotropic Gaussian distribution with $\frac{1}{2}\sigma$ equal to the bounding box diagonal.
We train our network until full convergence, which takes around 500k--900k iterations for each run, taking on average 6 hours/model.

%% file: fig/architecture.tex
\begin{figure}
\begin{center}
\includegraphics[width=\linewidth]{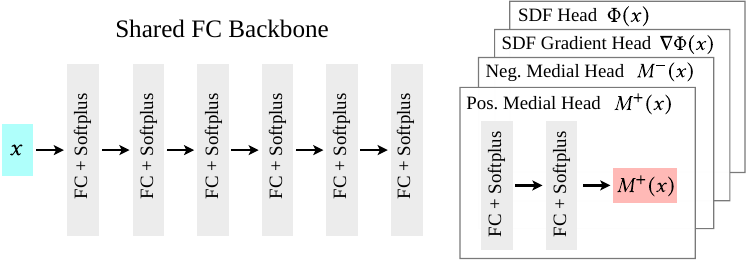}
\end{center}
\vspace{-1em}
\caption{
\textbf{Architecture} -- 
A block diagram of our network architecture. A shared fully connected backbone (left) embeds a world coordinate into a learned feature, which is then decoded into $\MField^+(x)$, $\SDF(x)$, $\MField^-(x)$ and $\nabla\SDF(x)$ by four separate, fully connected heads.
%
}
\label{fig:architecture}
\end{figure}

%% file: fig/surface_quality.tex
\begin{table}
\begin{center}
\begin{tabular}{r|cc}
\toprule
\textbf{2D}       & SDF & DMF \\ 
\midrule
giraffe   & \bf{0.018}          & 0.032          \\
koch      & \bf{0.051}          & 0.055           \\
m         & \bf{0.017}          & 0.024           \\
maple     & 0.016          & \bf{0.014}           \\
octopus   & \bf{0.020}          & 0.024           \\
statue    & \bf{0.017}          & 0.020          \\
\bottomrule
\end{tabular}
%
\hspace{1em}
%
\begin{tabular}{r|cc}
\toprule
\textbf{3D}       & SDF & DMF \\ 
\midrule
armadillo  & 0.045          & \bf{0.043}          \\
bunny      & 0.040          & \bf{0.032}          \\
horse      & \bf{0.022}          & 0.023          \\
lucy       & \bf{0.048}          & 0.051          \\
mecha      & \bf{0.071}          & 0.072          \\
rocker-arm & \bf{0.028}          & 0.035         \\
\bottomrule
\end{tabular}
\end{center}
\caption{
\textbf{Interference analysis} --
We analyze whether adding the losses needed to train DMF to classical SDF training degrades the surface approximation quality.
Values are mean absolute error expressed as a percentage of the bounding box diagonal.
}
\label{tab:surface_quality}
\end{table}


%% file: fig/hyperparam.tex
\begin{table}[t]
\begin{center}
\resizebox{\linewidth}{!}{
\setlength{\tabcolsep}{2pt}
\begin{tabular}{ccccccccc}
\toprule
$\Loss{surface}$ & $\Loss{normal}$ & $\Loss{maximal}$ & $\Loss{inscribed}$ & $\Loss{orthogonal}$ & $\Loss{Eikonal}$ & $\Loss{minsurface}$ & $\Loss{curvature}$ & $\Loss{gradient}$  \\
\midrule
$10^4$ & $10$ & $10^2$ & $5\times10^2$ & $3\times10^{-2}$ & $1$ & $1$ & $10^{-(1+4t)}$ & $1$ \\
\bottomrule
\end{tabular}
}
\end{center}
\caption{\textbf{Hyperparameter setup} -- Determined empirically to balance the range of terms. We also schedule the curvature loss term during training (denoted here with $t$, a value varying linearly from 0 to 1 over the course of training) in order to strongly regularize at the beginning to avoid local optimum, and allow a closer fit to the surface near the end.}
\label{tab:hyperparam}
\end{table}

%% file: sec/6_conclusions.tex
\section{Conclusions}
We have introduced medial fields, an implicit representation of the local thickness, that expands the capacity of implicit representations for 3D geometry.
We have shown how medial fields could be stored  within the parameters of a deep neural network, not only allowing O(1) access, but providing an effective strategy for its computation.
We have also demonstrated the potential of medial fields by showing their use in a number of applications: we show that it can be used to \CIRCLE{1} improve convergence when sphere tracing, \CIRCLE{2} efficiently provide physics proxies, and \CIRCLE{3} perform ambient occlusion.

\paragraph{Limitations and future works}
While our method is presented as dependent on \textit{signed} distance fields, this is only required to resolve the surface discontinuity in $\MField(x)$, and this requirement could be alleviated with an alternate approach to modelling the discontinuity.
Similarly to sphere tracing, medial sphere tracing can result in rendering artifacts when the underlying field values are not accurate.
This manifests in the case where queried ``empty'' spheres are not actually empty, resulting in the tracing stepping over the surface.
Similarly to sphere tracing, this can be mitigated by assuming that the field values are over-estimated and scaling down the spheres, but doing this excessively this will begin to erode the advantage of using medial spheres.
There have been other proposed modifications to sphere tracing which address pathological cases using heuristics~\cite{balint2018accelerating,korndorfer2014enhanced}, and it could be beneficial to combine these with the additional information from medial spheres to compound the benefit to convergence.
Coarse-to-fine rendering schemes could also potentially see significant improvement from medial spheres, as they do not approach zero radius near the surface, and would therefore delay subdivision of the ray packets.
Computing spherical physics proxies might not be optimal for geometry exhibiting very thin geometry, or be simply unsuitable for non-orientable or non-watertight surfaces~\cite{chibane2020ndf}.
Further, while our physics proxy algorithm is a heuristic, a venue for future work is the investigation of optimization-based techniques that are capable of building proxies with guarantees on the achieved approximation power~(e.g. bounded Hausdorff error), or optimal placement for a fixed cardinality of proxies.

\begin{acks}
The authors would like to thank Brian Wyvill, Frank Dellaert, and Ryan Schmidt for their helpful comments.
This work was supported by the Natural Sciences and Engineering Research Council of Canada (NSERC) Discovery Grant, NSERC Collaborative Research and Development Grant, Google, Compute Canada, and Advanced Research Computing at the University of British Columbia.
\end{acks}

%% file: sec/X_appendix.tex
\appendix
\section{Proof}
\label{app:proof}
In this section we prove the following (recall that $\USDF(x)=|\SDF(x)|$):
\begin{prop}
   Let $\MField^*$ be a function on $\Real^d \setminus \Boundary$ such that:
   \begin{align*}
     \MField^*(x) \geq \USDF(x)
     \quad &\forall x \in \Real^d \setminus \Boundary \;, \\
     \MField^*(x) = \USDF\left(x + (\MField^*(x)-\USDF(x)) \nabla \USDF(x)\right)
     \quad &\forall x \in \Real^d \setminus \Boundary \;, \\
     \nabla \MField^*(x) \cdot \nabla \SDF(x) = 0 
     \quad &\forall x \in \Real^d \setminus (\Boundary \cup \MAxis) \;,
  \end{align*}
  Then $\MField^*(x) = \MField(x)$ where $\MField$ is the
  median field of $\Shape$.
\end{prop}
\begin{proof}
 Let $\Line(x)$ be the spoke, i.e. the line segment connecting $\Project\Shape(x) = x_0$ to $\Project\MAxis(x)$.
 We can parametrize this line by
 \[ \Line(x) = \{ x_0 + t \nabla \USDF(x)\, | \, \text{ for } t \in \left(0, \MField(x)\right)\}\;. \]
 Consider $f(t) = \USDF(x_0 + t \nabla \USDF(x))$. Note that
 since $\USDF(x)$ is differentiable on $\Line(x)$, we can use chain rule to get $f'(t)=\|\nabla \USDF(x)\|^2=1$ for all $t \in \left(0, \MField(x)\right)$. Therefore  $\USDF(x_0 + t \nabla \USDF(x)) = t$ for $t \in (0, \MField(x))$.
 Similarly, let $g(t) = \MField^*(x_0 + t \nabla \USDF(x))$,
 and since we are asssuming $\nabla \MField^*(x) \cdot \nabla \SDF(x)=0$,
 we get $g'(t)=0$, which implies $\MField^*(x)$ is constant
 on $\Line(x)$.
 Recall that we're assuming that $\MField^*(x) \geq \USDF(x)$. Restricting this inequality to $\Line(x)$
 we get
 \[ \MField^*(x) = \MField^*(x_0 + t \nabla \USDF(x)) \geq \USDF(x_0 + t \nabla \USDF(x)) = t\;,\]
 for all $t \in (0, \MField(x))$; therefore we get $\MField^*(x) \geq \MField(x)$.
 Finally, note that $x_0 + \MField(x) \nabla \USDF(x) \in \MAxis$. If $\MField(x)>\MField^*(x)$ then,
 using the assumption that $\MField^*(x) = 
 \USDF(x + (\MField^*(x)-\USDF(x)) \nabla \USDF(x))$:
 \begin{align*}
     \MField^*(x) & = \USDF\left(x + (\MField^*(x)-\USDF(x)) \nabla \USDF(x)\right) \\
     & = \USDF\left(x_0 + \MField^*(x) \nabla \USDF(x)\right) \\
     & = \USDF\left(x_0 + \MField(x) \nabla \USDF(x) + (\MField^*(x) - \MField(x)) \nabla \USDF(x)\right) \\
     & < \USDF\left(x_0 + \MField(x) \nabla \USDF(x)\right) + \| (\MField^*(x) - \MField(x))\nabla \USDF(X)\| \\
     & = \MField^*(x)\;,
    \end{align*}
which is a contradiction. Therefore $\MField(x) = \MField^*(x)$ as desired.
\end{proof}